%% file: tchond-corr-2018.tex
\documentclass[sigconf]{acmart}

\usepackage{booktabs} 
\usepackage{amsmath}
\usepackage{xspace}
\usepackage{pgfplots}
\usepackage{xcolor}
\usepackage{colortbl}
\usepackage{caption}
\usepackage{subcaption}
\usepackage{multirow}
\usepackage{enumerate}
\usepackage[ruled,linesnumbered,vlined,ruled]{algorithm2e} 

\usetikzlibrary{patterns}

\newcommand{\eat}[1]{}
\renewcommand{\sp}{\ensuremath{p_{sp}}} 

\newcommand{\kdsp}{\ensuremath{k\text{DPwML}}\xspace}
\newcommand{\kspwlo}{\ensuremath{k\text{SPwLO}}\xspace}
\newcommand{\olFunc}{Sim\xspace}
\newcommand{\Q}{\ensuremath{\mathcal{Q}}\xspace}

\newcommand{\PALL}{\ensuremath{P_{all}}\xspace}
\newcommand{\PSP}{\ensuremath{P_{sp}}\xspace}
\newcommand{\PSSV}{\ensuremath{P_{ssv}}\xspace} 
\newcommand{\PDSP}{\ensuremath{P_{\kdsp}}\xspace}

\DeclareMathOperator*{\poly}{poly}

\newcommand{\len}{\ensuremath{\ell}}
\newcommand{\LEN}{\ensuremath{\mathcal{L}}}

\newtheorem{problem}{Problem}

\newcounter{myfunction}
\newenvironment{myfunction}[1][t]{
	
    \begin{algorithm}[#1]%
}{
    \end{algorithm}
}

\newcommand{\kspdml}{$K$SP-DML\xspace} 

\newcommand{\kspdplus}{FindKSPD\xspace}
\newcommand{\svpdml}{SSVP-DML\xspace} 

\newcommand{\svpdplus}{SSVP-D+\xspace}

\renewcommand\footnotetextcopyrightpermission[1]{} 
\begin{document}
\title{Finding $k$-Dissimilar Paths with Minimum Collective Length}
\titlenote{Extended version of the ACM SIGSPATIAL'18 paper under the same title.}

\author{Theodoros Chondrogiannis}
\affiliation{%
  \institution{
  University of Konstanz}
}
\email{theodoros.chondrogiannis@uni.kn}

\author{Panagiotis Bouros}
\affiliation{%
  \institution{
  Johannes Gutenberg University Mainz}
}
\email{bouros@uni-mainz.de}

\author{Johann Gamper}
\affiliation{%
  \institution{
  Free University of Bozen-Bolzano}
}
\email{gamper@inf.unibz.it}

\author{Ulf Leser}
\affiliation{%
  \institution{
  Humboldt-Universit{\"a}t zu Berlin}
}
\email{leser@informatik.hu-berlin.de}

\author{David B. Blumenthal}
\affiliation{%
  \institution{
  Free University of Bozen-Bolzano}
}
\email{david.blumenthal@inf.unibz.it}

\renewcommand{\shortauthors}{T. Chondrogiannis et al.}

\begin{abstract}
Shortest path computation is a fundamental problem in road networks. However, in
many real-world scenarios, determining solely the shortest path is not enough.
In this paper, we study the problem of finding $k$-Dissimilar Paths with Minimum
Collective Length (\kdsp), which aims at computing a set of paths from a source $s$
to a target $t$ such that all paths are pairwise dissimilar by at least $\theta$ and
the sum of the path lengths is minimal. We introduce an exact algorithm for the
\kdsp problem, which iterates over all possible $s{-}t$ paths while employing two
pruning techniques to reduce the prohibitively expensive computational cost. To
achieve scalability, we also define the much smaller set of the simple single-via
paths, and we adapt two algorithms for \kdsp queries to iterate over this set.
Our experimental analysis on real road networks shows that iterating over all paths
is impractical, while iterating over the set of simple single-via paths can lead to
scalable solutions with only a small trade-off in the quality of the results.
\end{abstract}

\keywords{Alternative Routing, Route Planning, Path Similarity}

\maketitle

\input{chapters/intro}

\input{chapters/related}
\input{chapters/preliminaries}
\input{chapters/exact}

\input{chapters/heuristic}
\input{chapters/experiments}

\input{chapters/conclusion}

\bibliographystyle{plain}
\bibliography{related} 

\end{document}

%% file: chapters/intro.tex
\section{Introduction}

Computing the shortest path between two locations in a road network is a fundamental problem
that has attracted the attention of both the research community and the industry. In many real-world
scenarios though, determining solely the shortest path is not enough. Most commercial route planning
applications recommend alternative paths that might be longer than the shortest path, but have other
desirable properties, e.g., less traffic congestion. However, the recommended paths also need to be
dissimilar to each other to be valued as true alternatives by users. Towards this end, various
approaches have been proposed that aim at computing short yet dissimilar to each other alternative
paths~\cite{akgun2000,chondrogiannis2015,chondrogiannis2017,jeong2009}.

\begin{figure}[htb]
\centering
\includegraphics[width=\columnwidth]{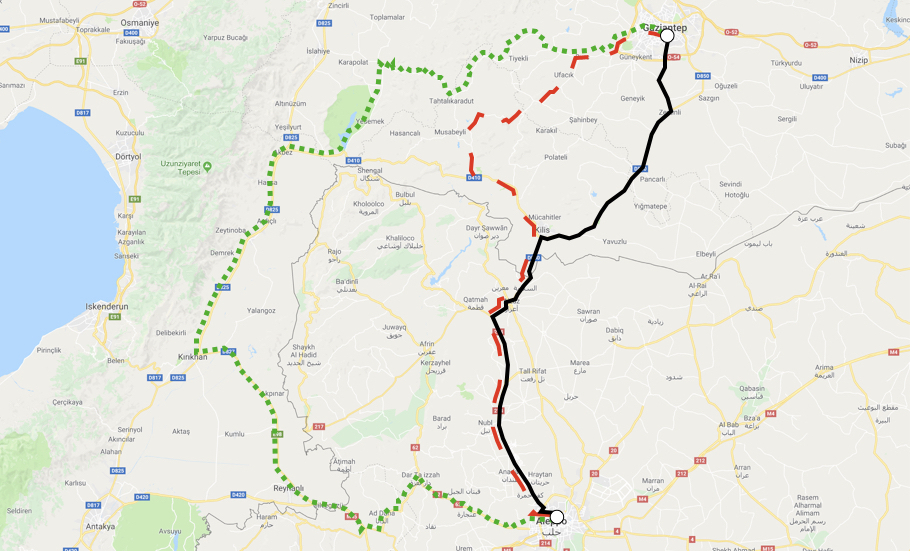}
\caption{Motivating example.}
\label{fig:example}
\end{figure}

In many real-world scenarios though, apart from ensuring the diversity of the recommended routes,
the collective length, i.e., the total distance covered by vehicles, must also be taken into account
to minimize the overall cost. Consider the scenario of transportation of humanitarian aid goods
through unsafe regions. The distribution of the load to several vehicles that follow routes dissimilar
can increase the chances that at least some of the goods will be delivered. The total distance
covered by the vehicles must also be taken into account to minimize the overall cost. For example,
Figure~\ref{fig:example} shows three different paths from the city of Gaziantep in Turkey to the city
of Aleppo in Syria. The solid/black line indicates the shortest path, the dashed/red line the next
path in length order, and the dotted/green line a path that is clearly longer, but also significantly
different from the other two. Choosing the black and the red paths is not the best option, since the
two paths share a large stretch. Among the other two options, the black-green pair has the minimum 
collective length and, hence, it is a better option than the red-green pair.

The aforementioned scenario is formally captured by the \emph{$k$-Dissimilar Paths with Minimum
Collective Length} (\kdsp) problem. Given two locations $s$ and $t$ on a road network, a \kdsp query
computes a set of $k$ paths from $s$ to $t$ such that: (1) all paths in the result set are sufficiently
dissimilar to each other (w.r.t.\ a user-defined similarity threshold), and (2) the set exhibits the
lowest collective path length among all sets of $k$ sufficiently dissimilar paths. \kdsp was originally
introduced by Liu et al.~\cite{liu2017} as the \emph{Top-$k$ Shortest Paths with Diversity} (Top-KSPD),
together with a greedy heuristic method that builds on the $K$-shortest paths~\cite{yen1971}.

In this paper, we present an in-depth analysis of the \kdsp problem. First, we conduct a theoretical
analysis to prove that \kdsp is strongly $NP$-hard. Second, we investigate the exact computation of
\kdsp queries, which was not covered by Liu et al.~\cite{liu2017}. We present an algorithm that,
similar to the approach of Liu et al., builds on the computation of the $K$-shortest paths~\cite{yen1971},
along with a pair of pruning techniques. Since such approaches require a prohibitively high number of
paths to be examined, we introduce the much smaller set of simple single-via paths, which extends the
concept of single-via paths~\cite{abraham2013}. Then, we present two algorithms that iterate over this
set of paths to compute \kdsp queries. Our experiments show that algorithms which iterate over all
possible $s{-}t$ paths cannot scale. Instead, iterating over the set of simple single-via paths can
lead to scalable solutions with a very small trade-off in the quality of the results. 

The rest of the paper is organized as follows.
In Section~\ref{sec:related} we discuss the related work. In Section~\ref{sec:preliminaries}
we introduce the necessary notation, we formally define the \kdsp problem and we show that
the problem is strongly $NP$-hard.
In Section~\ref{sec:kdsp} we present our exact algorithm along with two pruning techniques.
In Section~\ref{sec:heuristic} we introduce the concept of simple single-via paths and we present
two heuristic algorithms that use simple single-via paths to evaluate \kdsp queries.
Finally, in Section~\ref{sec:exps} we present the results of our experimental evaluation,
and Section~\ref{sec:conclusion} concludes this work.

%% file: chapters/related.tex
\section{Related Work}\label{sec:related}

Different forms of alternative routing have been proposed in the past. Liu et al.~\cite{liu2017} introduced
the problem of finding Top-$k$ Shortest Paths with Diversity ($\mathit{Top{-}KSPD}$), which we study in this
paper as the \kdsp. Despite introducing the problem though, Liu et al. investigated only the approximate
computation of \kdsp queries, and they proposed the greedy heuristic algorithm \kspdplus, that builds
upon the computation of the $K$-Shortest Paths~\cite{yen1971}, while employing two pruning criteria to limit the
number of examined paths. According to the first criterion, partially expanded paths to the same node
$n$ are grouped together and the most promising ones are expanded first. The second pruning criterion
involves the computation of lower bounds either based on the estimated length or on the similarity of
paths to prioritize the examination of paths that are more likely to lead to a solution.

Another way to compute dissimilar alternative routes is to solve the $k$-Shortest Paths with Limited
Overlap (\kspwlo) problem, introduced by Chondrogiannis et al.~\cite{chondrogiannis2015}. A \kspwlo query
aims at computing paths that are (a) the shortest path is always included and (b) every new path added to the
set is alternative to all the shorter paths already in the set and as short as possible. The authors proposed
both exact and heuristic algorithms~\cite{chondrogiannis2015,chondrogiannis2017} to process \kspwlo queries,
using the path overlap~\cite{abraham2013} as similarity measure. In practice, the greedy
\kspdplus algorithm introduced by Liu et al.~\cite{liu2017} for the \kdsp is, unbeknownst to the authors,
an exact algorithm for the \kspwlo for arbitrary similarity measures, similar to the baseline approach of 
Chondrogiannis et al.~\cite{chondrogiannis2015}. Consequently, the \kspwlo can be seen as an approximation
to the \kdsp.

A similar approach to \kspwlo has been proposed by Jeong et al.~\cite{jeong2009}, that is laxly based on Yen's
algorithm~\cite{yen1971}. At each step, the algorithm modifies each previously computed path to construct a
set of candidate paths and examines the one that is most dissimilar to the already recommended paths. 
Akgun et al.~\cite{akgun2000} proposed a penalty-based method that doubles the weight of each edge that lies
on some already recommended path. The alternative paths are computed by repeatedly running Dijkstra's algorithm
on the input road network, each time with the updated weights. The shortcoming of this approach is that there
is no intuition behind the value of the penalty applied before each subsequent iteration. Lim et al. proposed
a similar penalty-based approach~\cite{lim2005} where the penalty is computed in terms of both the path overlap
and the total turning cost, i.e., how many times the user has to switch between roads when following a path.
In contrast to the \kdsp and the \kspwlo, neither approach of Jeong et al.\ nor the penalty-based methods 
of Akgun et al.\ and Lim et al.\ come with an optimization criterion w.r.t.\ the length of the alternative paths.

A different approach to alternative routing involves methods that focus on computing alternatives
only to the shortest paths. Such methods first compute a large set of candidate paths, and then determine
the final result set by examining the paths with respect to a number of user-defined constraints.
The \emph{Plateaux} method~\cite{cambridge2005} aims at computing paths that cross different highways
of the road network. Bader et al.~\cite{bader2011}, introduced the concept of alternative graphs which
have a similar functionality as the plateaus. Abraham et al.~\cite{abraham2013} introduced the notion
of \emph{single-via paths}, a set that we extend for developing our heuristic algorithms. The approach
of Abraham et al. evaluates each single-via path individually by comparing it to the shortest path, and
checks whether each path meets a set of user-defined constraints, i.e., length, local optimality and
stretch. Xie et al.~\cite{xie2012} define alternative paths to the shortest path using edge avoidance, i.e.,
given an edge of $e$ of the shortest path from $s$ to $t$, the alternative path is the shortest path from
$s$ to $t$ which avoids $e$, and introduce iSQPF, a quadtree-based spatial data structure inspired by~\cite{samet2008}. 
In contrast to our work, none of the aforementioned methods aims at minimizing the collective
length of the result set, or guarantees that the result paths will be dissimilar to each other.

Finally, the task of alternative routing can also be based on the pareto-optimal paths for multi-criteria
networks~\cite{delling2009,kriegel2010,mouratidis2010,shekelyan2015}. The pareto-optimal paths or the route
skyline can be directly seen as alternative routes to move from source node $s$ to target node $t$ or can be
further examined in a post-processing phase to provide the final alternative paths. However, our definition of
alternative routing is not a multi-criteria problem and the \kspwlo cannot be obtained by first computing the
pareto-optimal path set.

%% file: chapters/preliminaries.tex
\section{Preliminaries} \label{sec:preliminaries}

Section~\ref{sec:preliminaries-notation} introduces the necessary notation while
Section~\ref{sec:preliminaries-kspd} formally defines the \kdsp problem and
Section~\ref{sec:preliminaries-complexity} analyzes its complexity.

\subsection{Notation} \label{sec:preliminaries-notation}

Let $G = (N,E)$ be a \emph{directed weighted} graph representing a road network with a set
of nodes $N$ and a set of edges $E\subseteq N \times N$.\footnote{For ease of presentation,
  we draw a road network as an undirected graph in our examples. However, our proposed methods
  directly work on directed graphs as well.}
Each edge $(n_i,n_j) \in E$ is assigned a \emph{positive} weight $w(n_i,n_j)$, which captures
the cost of moving from node $n_i$ to node $n_j$. A (simple) \emph{path} $p(s{\rightarrow}t)$
from a source node $s$ to a target node $t$ is a \emph{connected} and \emph{cycle-free} sequence
of edges $\langle (s,n_i),\ldots,(n_j,t)\rangle$. The \emph{length} $\len(p)$ of a path $p$ is
the sum of the weights of all contained edges, and the \emph{collective path length} $\mathcal{L(P)}$
of set of paths $P$ is the sum of the lengths of the paths in the set.The \emph{length} $\len(p)$
of a path $p$ is the sum of the weights of all contained edges, i.e.,
\begin{equation*}
  \len(p) = \sum_{\forall (n_i,n_j) \in p} w(n_i,n_j)
\end{equation*}
We denote by $\mathcal{L}$ the \emph{collective path length} for a set of paths $P$, i.e.,
\[
\mathcal{L}(P) = \sum_{\forall p \in P}{\ell(p)}
\]
The \emph{shortest path} $\sp(s{\rightarrow}t)$ is the path with the lowest length among all
paths that connect $s$ to $t$. 

Let $p$, $p'$ be two paths between nodes $s$, $t$. We denote the \emph{similarity} of the paths
as $\olFunc(p,p')$.  Given a similarity threshold $\theta$, we say that paths $p$, $p'$ are
\emph{sufficiently dissimilar} if $\olFunc(p,p') < \theta$. We also say that a path $p$ is
sufficiently dissimilar to a set of paths $P$ w.r.t.\ a threshold $\theta$, if $p$ is sufficiently
dissimilar with every path in $P$.

In the past, various measures have been proposed to compute path similarity (cf.~\cite{liu2017}).
Choosing the proper similarity measure heavily depends on the application and hence is out of the
scope of our work. Nevertheless, the algorithms we present operate with any arbitrary similarity
measure. Hence, without loss of generality, we use the Jaccard coefficient, i.e.,
\[\olFunc(p,p') = \frac{\sum_{\forall (n_i,n_j) \in p \cap p'}w(n_i,n_j)}{\sum_{\forall (n_i,n_j) \in p \cup p'}w(n_i,n_j)}\]
Table~\ref{tab:notation} summarizes the notation used throughout this paper.

\begin{table}[t]
	\small
	\centering
	\caption{Summary of notation.\label{tab:notation}}
	\begin{tabular}{|l|l|}
    	\hline
    	\textbf{Notation} & \textbf{Description} \\
    	\hline\hline
    	$G = (N,E)$ & Graph $G$ with nodes $N$ and edges $E$\\
    	$n$ & Node in $N$\\
    	$(n_i,n_j)$ & Edge from node $n_i$ to node $n_j$\\
    	$w(n_i,n_j)$ & Weight of edge $(n_i,n_j)$\\
    	$p = \langle (s,n_1),\dots,(n_k,t) \rangle$ & Path $p$ from node $s$ to node $t$\\
    	$\len(p)$ & Length of path $p$\\
    	$\LEN(P)$ & Collective length for a set of paths $P$\\
    	$\olFunc(p,p')$ & Similarity between two paths $p$ and $p'$\\ 
    	\hline
	\end{tabular}
\end{table}

\begin{example}
Consider the road network in Figure~\ref{fig:altrouting-alternativeex}. The shortest path between
nodes $s$ and $t$ is   $p_1{=}p_{sp}{=}\langle(s,n_3),(n_3,n_5),(n_5,t)\rangle$ with length
$\len(p_1) = 8$. Let $p_2=\langle(s,n_3),(n_3,n_5),(n_5,n_4),(n_4,t)\rangle$ and
$p_3{=}\langle(s,n_3),$ $(n_3,n_4),(n_4,t)\rangle$ be two more paths   that connect $s$ to $t$
with length $\len(p_2){=}9$ and $\len(p_3){=}10$, respectively. Assuming a similarity threshold
$\theta{=}0.5$, paths $p_{1}$ and $p_2$ are not sufficiently dissimilar to each other as their
similarity exceeds $\theta$, i.e., $\olFunc(p_1,p_2){=}6/11 \approx 0.545 >\theta$ due to the shared
edges $(s,n_3)$ and $(n_3,n_5)$. In contrast, paths $p_3$ and $p_1$ that share only edge $(s,n_3)$
are sufficiently dissimilar, i.e., $\olFunc(p_1,p_3){=}3/15{=}0.2 < \theta$.
\end{example}

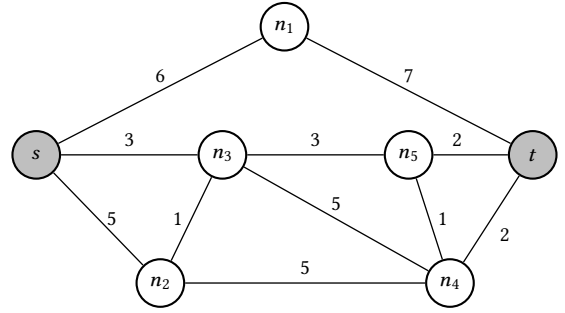
\begin{figure}[t]
\centering
\begin{tikzpicture}[
		auto,
		node distance=1cm,
  		thick,
  		inner/.style={
	  		circle,
  			draw,
  			font=\small,
  			minimum size = 18pt,
  			inner sep=0pt
  		},
  		point/.style={
  			circle,
  			draw,
  			font=\small,
  			fill=lightgray,
  			minimum size = 18pt,
  			inner sep=0pt
  		},
  		box/.style={
  			rectangle,
  			draw,
  			minimum width = 4cm,
  			font=\tiny
  		},
  		scale=1.1
	]
	\input{tikz/basegraph.tikz}
\end{tikzpicture}
\caption{Running example.}
\label{fig:altrouting-alternativeex}
\end{figure}

\subsection{Problem Definition} \label{sec:preliminaries-kspd}

We now formally restate the problem of finding \emph{Top-$k$ Shortest Paths with Diversity}~\cite{liu2017}
as the \emph{$k$-Dissimilar Paths with Minimum (Collective) Length (\kdsp)}, using the terminology of
Section~\ref{sec:preliminaries-notation}.
\begin{problem}[\kdsp]
Given a road network $G = (N,E)$, a source $s$ and a target node $t$ both in $N$, a number of requested
paths $k$, and a similarity threshold $\theta$, find the \PDSP set of paths from $s$ to $t$, such that:
\begin{enumerate}[(A)]
\item all paths in $\PDSP$ are pairwise sufficiently dissimilar,
\[
\forall p_i,p_j \in \PDSP \text{, } i \neq j : Sim(p_i,p_j) < \theta,
\]
\item $\vert \PDSP \vert \leq k$ and \PDSP has the maximum possible cardinality among every set of paths
$P_A$ that satisfy Condition (A), 
\[|
\PDSP| = \arg \max_{P_A} |P_A|\text{, and}
\]
\item \PDSP has the lowest collective path length among every set of paths $P_{AB}$ that satisfy both
Conditions (A) and (B),
\[
\PDSP = \arg \min_{P_{AB}} \LEN(P_{AB}).
\]
\end{enumerate}
\label{def:kdsp}
\end{problem}
\noindent Intuitively, a \kdsp query returns the maximal set of at most $k$ sufficiently dissimilar paths
w.r.t. threshold $\theta$, which have the lowest collective length among all sets of sufficiently dissimilar
$s{\rightarrow}t$ paths.
\begin{example}
Consider again the example in Figure~\ref{fig:altrouting-alternativeex} and paths
\begin{align*}
p_1 = p_{sp} & = \langle(s,n_3), (n_3,n_5), (n_5,t)\rangle \\
p_2 & = \langle (s,n_3),(n_3,n_5),(n_5,n_4),(n_4,t) \rangle \\
p_3 & = \langle (s,n_3),(n_3,n_4),(n_4,t) \rangle \\
p_4 & = \langle (s,n_2),(n_2,n_3),(n_3,n_5),(n_5,t) \rangle
\end{align*}
Let $P_1 =\{p_1,p_2,p_3\}$, $P_2 = \{p_1,p_3,p_4\}$, and $P_3 = \{p_2,p_3,p_4\}$ be three sets of paths
with $\LEN(P_1) = 27$, $\LEN(P_2) = 29$, and $\LEN(P_3) = 30$. Consider the query $\kdsp(s,t,3,0.5)$. 
While set $P_1$ has the lowest collective length, it cannot be the result set as $p_1$ and $p_2$ are not
sufficiently dissimilar, i.e., $\olFunc(p_1,p_2) = 6/11 = 0.545 > \theta$. On the other hand, both $P_2$
and $P_3$ contain sufficiently dissimilar paths, but $P_2$ is preferred as $\LEN(P_2) < \LEN(P_3)$.
In fact, $P_2$ is the set with the lowest collective length among all sets that contain three sufficiently
dissimilar paths. Hence, $P_2$ is the result of the query.
\end{example}

\subsection{Complexity} \label{sec:preliminaries-complexity}

We now elaborate on the complexity of the \kdsp problem. Liu et al.\ proved in~\cite{liu2017}
(cf.\@ Lemma~1) the $NP$-hardness of \kdsp. Despite the correctness of their finding, the authors'
approach on the proof is incorrect as they polynomially reduced \kdsp \emph{to} a hard problem, i.e.,
the Maximum Independent Set problem, instead of providing a polynomial reduction \emph{from} a hard
problem~\cite{arora2009}. In view of this, we hereby prove the following theorem.
\begin{theorem}
The \kdsp problem is strongly $NP$-hard.
\end{theorem}
\begin{proof}
We prove the lemma by polynomial reduction from the two edge-disjoint path problem in directed graphs
($2$-DP), which is known to be strongly $NP$-complete~\cite{eilam-tzoreff1998}. Given a directed graph
$G=(N,E)$ with $|N|=n$ and two source-target pairs $\{s_1,t_1\}$ and $\{s_2,t_2\}$, $2$-DP asks to correctly
decide if $G$ contains edge-disjoint  paths from $s_i$ to $t_i$ for $i=1,2$. For polynomially reducing
$2$-DP to the task of answering \kdsp queries, we we define a road network $G^\prime=(N^\prime,E^\prime)$
with $N^\prime=N\cup\{s,t,a,b,c,d\}$ and
$E^\prime=E\cup\{(s,a),(s,c),(a,b),(c,d),(b,s_1),(d,s_2),(t_2,a),(t_1,c),(b,t),(d,t)\}$.
We set $k=4$, $\theta=(2n+1)/(4n^2+5)$, $w(e)=4n$ for all $e \in E$, and $w(e)=1$ for all
$e\in E^\prime\setminus E$. Cf.\@ Figure~\ref{fig:proof} for a visualization of the network $G^\prime$.
We claim that there are edge-disjoint paths from $s_i$ to $t_i$ in $G$ just in case the result of a
$4$-DPwML query against $G^\prime$ has cardinality $4$. This implies that, unless $P=NP$, there can be
no polynomial or pseudo-polynomial algorithm for answering \kdsp queries.
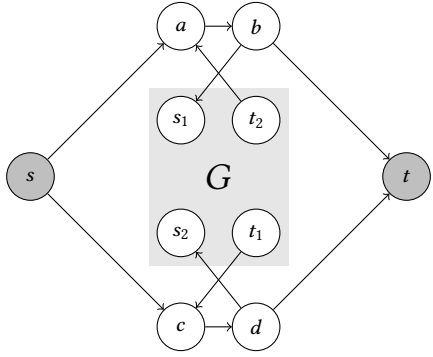
\begin{figure}
\usetikzlibrary{fit,backgrounds}
\begin{tikzpicture}[
inner/.style={
	circle,
  	draw,
  	font=\small,
  	minimum size = 18pt,
  	inner sep=0pt
  	},
  	point/.style={
  		circle,
  		draw,
  		font=\small,
  		fill=lightgray,
  		minimum size = 18pt,
  		inner sep=0pt
}]
\node[point] (s) at (-1,0) {$s$};
\node[inner] (a) at (1,2) {$a$};
\node[inner] (b) at (2,2) {$b$};
\node[inner] (c) at (1,-2) {$c$};
\node[inner] (d) at (2,-2) {$d$};
\node[point] (t) at (4,0) {$t$};
\node[inner,fill=white] (s-1) at (1,.75) {$s_1$};
\node[inner,fill=white] (s-2) at (1,-.75) {$s_2$};
\node[inner,fill=white] (t-2) at (2,.75) {$t_2$};
\node[inner,fill=white] (t-1) at (2,-.75) {$t_1$};
\node (G) at (1.5,0) {\huge $G$};
\begin{scope}[on background layer]
\node[fit=(s-1) (s-2) (t-1) (t-2),fill=black!10] {};
\end{scope}
\draw[->] (s) -- (a);
\draw[->] (a) -- (b);
\draw[->] (s) -- (c);
\draw[->] (c) -- (d);
\draw[->] (b) -- (s-1);
\draw[->] (t-2) -- (a);
\draw[->] (d) -- (s-2);
\draw[->] (t-1) -- (c);
\draw[->] (b) -- (t);
\draw[->] (d) -- (t);
\end{tikzpicture}
\caption{The network $G^\prime$ used for proving that \kdsp is $NP$-hard. All visualized edges
have weight $1$ and the edges contained in the sub-graph $G$ have weight $4n$, where $n$ is the
size of $G$.}\label{fig:proof}
\end{figure}

For proving the claim, we make use of the fact that the similarity of paths $p^\prime_i(s{\rightarrow}t)$
and $p^\prime_j(s{\rightarrow}t)$ can be expressed as follows:
\begin{equation}
\olFunc(p^\prime_i,p^\prime_j)=\frac{\overbrace{\sum_{e\in(p^\prime_i\cap p^\prime_j)\cap E}w(e)}^{A(p^\prime_i,p^\prime_j)}+\overbrace{\sum_{e\in(p^\prime_i\cap p^\prime_j)\setminus E}w(e)}^{B(p^\prime_i,p^\prime_j)}}{\underbrace{\sum_{e\in(p^\prime_i\cup p^\prime_j)\cap E}w(e)}_{C(p^\prime_i,p^\prime_j)}+\underbrace{\sum_{e\in(p^\prime_i\cup p^\prime_j)\setminus E}w(e)}_{D(p^\prime_i,p^\prime_j)}}
\end{equation}

We first assume that there are edge-disjoint paths $p_1=\langle s_1,\ldots,t_1\rangle$ and
$p_2=\langle s_2,\ldots,t_2\rangle$ in $G$. Consider the following paths in $G^\prime$
$p^\prime_1=\langle s,a,b,p_1,c,d,t\rangle$, $p^\prime_2=\langle s,c,d,p_2,a,b,t\rangle$,
$p^\prime_3=\langle s,a,b,t\rangle$, and $p^\prime_4=\langle s,c,d,t\rangle$.
Then the following statements immediately follow from the construction of $G^\prime$:

\begin{enumerate}
\item Since paths $p_1$ and $p_2$ are edge disjoint, we have $A(p^\prime_i,p^\prime_j)=0$ for all
$(i,j)\in\{1,\ldots,4\}\times\{1,\ldots,4\}$.
\item We have $B(p^\prime_3,p^\prime_4)=0$ and $B(p^\prime_i,p^\prime_j)=2$ for all
$(i,j)\in\{1,\ldots,4\}\times\{1,\ldots,4\}\setminus\{(3,4),(4,3)\}$.
\item Since the paths $p_1$ and $p_2$ use at least one edge in $E$, we have $C(p^\prime_i,p^\prime_4)\geq4n$
for all $(i,j)\in\{1,\ldots,4\}\times\{1,\ldots,4\}\setminus\{(3,4),(4,3)\}$.
\item We have $D(p^\prime_i,p^\prime_4)\geq7$ for all
$(i,j)\in\{1,\ldots,4\}\times\{1,\ldots,4\}\setminus\{(3,4),(4,3)\}$.
\end{enumerate}

These statements imply that, for all $(i,j)\in\{1,\ldots,4\}\times\{1,\ldots,4\}$, we have
$\olFunc(p^\prime_i,p^\prime_j)\leq2/(4n+7)<(2n+1)/(4n^2+5)=\theta$, and hence, that
$\PDSP=\{p^\prime_i\}^4_{i=1}$ is a result set for $4$-DPwML with cardinality four. 

For the other direction of the claim, assume that there is a set $\PDSP=\{p^\prime_i\}^4_{i=1}$
of four paths from $s$ to $t$ in $G^\prime$ with pairwise similarity less than $\theta=(2n+1)/(4n^2+5)$.
By construction of $G^\prime$, we know that each path either starts with the prefix
$\mathit{prefix}_1=\langle (s,a),(a,b)\rangle$ or with the prefix
$\mathit{prefix}_2=\langle (s,c),(c,d) \rangle$. Furthermore, we can see that a path
$p^\prime_i(s{\rightarrow}t)$ that starts with $\mathit{prefix}_1$ either equals
$\langle (s,a),(a,b),(b,t) \rangle$ or, at node $b$, enters $G$ through the edge $(b,s_1)$.
Assume that the latter is the case. Then we know that $p^\prime_i$ exits $G$ through the edge
$(t_1,c)$, since exiting $G$ through $(t_2,a)$ would close the cycle
$\langle (a,b),(b,s_1) (s_1,\ldots),(\ldots,t_2), (t_2,a) \rangle$ and hence contradict the
fact that $p^\prime_i$ is a path. Therefore, $p^\prime_i$ starts with the prefix
$\mathit{prefix}_1=\langle (s,a),(a,b),(b,s_1), \ldots,(t_1,c),(c,d) \rangle$. At node $d$,
$p^\prime_i$ cannot enter $G$ again, since both exit edges $(t_1,c)$ and $(t_2,a)$ would
close a cycle. This implies that $p^\prime_i=\langle s,a,b,s_1,\ldots,t_1,c,d,s\rangle$.
Analogously, we can show that a path $p^\prime_i(s{\rightarrow}t)$ that starts with
$\mathit{prefix}_2$ either equals $\langle (s,c),(c,d),(d,t) \rangle$ or is of the form
$p^\prime_i=\langle (s,c),(c,d),(d,s_2),(s_2,\ldots),(\ldots,t_2),(t_2,a),(a,b),(b,s) \rangle$.

These observations imply that $\PDSP$ contains paths
$p^\prime_1=\langle (s,a),(a,b) \rangle \cdot p_1 \cdot \langle (c,d) (d,s) \rangle$,
$p^\prime_2=\langle (s,c),(c,d) \rangle \cdot p_2 \cdot \langle (a,b),(b,s) \rangle$,
$p^\prime_3=\langle (s,a),(a,b),(b,t) \rangle$, and
$p^\prime_4=\langle (s,c),(c,d),(d,t) \rangle$, where $p_i$ are paths from some $s_i$ to some
$t_i$ in $G$. It remains to be shown that $p_1$ and $p_2$ are edge-disjoint. Assume that this
is not the case. Then the following statements hold:
\begin{enumerate}
\item Since $p_1$ and $p_2$ share at least one edge from $E$, we have
$A(p^\prime_1,p^\prime_2)\geq4n$.
\item Since $p^\prime_1$ and $p^\prime_2$ share exactly two edges from $E^\prime\setminus E$,
we have $B(p^\prime_1,p^\prime_2)=2$.
\item Since both $p_1$ and $p_2$ contain at most $n$ edges from $E$,
we have $C(p^\prime_1,p^\prime_2)\leq8n^2$.
\item Since each edge from $E^\prime\setminus E$ is contained in $p^\prime_1$ or in $p^\prime_2$,
we have $D(p^\prime_1,p^\prime_2)=10$.
\end{enumerate}
This statements imply that
\[
\olFunc(p^\prime_1,p^\prime_2)\geq(4n+2)/(8n^2+^0)=(2n+1)/(4n^2+5)=\theta
\]
which contradicts the fact that $\olFunc(p^\prime_1,p^\prime_2)<\theta$. Therefore, $p_1$
and $p_2$ are edge-disjoint, which yields the claim and finishes the proof of the theorem.
\end{proof}

%% file: tikz/basegraph.tikz
  	\node[point]	(N0)	at (-3,1.75)		{$s$};
  	\node[inner]	(N1)	at (0,3.3)			{$n_1$};
  	\node[inner]	(N2)	at (-1.5,0.2)		{$n_2$};
  	\node[inner]	(N3)	at (-0.75,1.75)		{$n_3$};
	\node[inner]	(N4)	at (2,0.2)			{$n_4$};
	\node[inner]	(N5)	at (1.5,1.75)		{$n_5$};
	\node[point]	(N6)	at (3,1.75)			{$t$};
%
%
	\path[line width=0.5pt, every node/.style={font=\small}]
	   	(N0) edge node[above]	{6} (N1)
		(N0) edge node[right]	{5} (N2)
		(N0) edge node[above]	{3} (N3)
		(N1) edge node[above]	{7} (N6)
		(N2) edge node[left]	{1} (N3)
		(N2) edge node[above]	{5} (N4)
		(N3) edge node[above]	{5} (N4)
		(N3) edge node[above]	{3} (N5)
		(N4) edge node[right]	{1} (N5)
		(N4) edge node[below right]	{2} (N6)
		(N5) edge node[above left]	{2} (N6);

%% file: chapters/exact.tex
\section{An Exact Approach} \label{sec:kdsp}

\begin{figure*}[t]
\centering
\begin{subfigure}[t]{0.31\textwidth }
\centering
\scalebox{1.1} {
\begin{tikzpicture}[
		auto,
		node distance=1cm,
 		thick,
  		point/.style={
  			circle,
  			draw,
  			font=\small,
  			fill=black,
  			minimum size = 5pt,
  			inner sep=0pt,
  			color=black
  		},
  		myline/.style={
  			line width=0.5pt,
  			every node/.style={font=\scriptsize},
  			color=black
  		},
	]
	\input{tikz/bintree.tikz}
\end{tikzpicture}
}
\caption{Subsets of $P = \{p_1,p_2,p_3\}$.}
\end{subfigure}
\hspace{2ex}
\begin{subfigure}[t]{0.55\textwidth }
\centering
\scalebox{1.1} {
\begin{tikzpicture}[
		auto,
		node distance=1cm,
 		thick,
  		point/.style={
  			circle,
  			draw,
  			font=\small,
  			fill=black,
  			minimum size = 5pt,
  			inner sep=0pt,
  			color=black
  		},
  		myline/.style={
  			line width=0.5pt,
  			every node/.style={font=\scriptsize},
  			color=black
  		},
  		subtree/.style={
  			line width=1.5pt,
  			every node/.style={font=\scriptsize},
  			color=black
  		},
  		mylinered/.style={
  			dashed,
  			line width=1.5pt,
  			every node/.style={font=\scriptsize}
    	},
    	mylineredthin/.style={
  			dashed,
  			line width=0.5pt,
  			every node/.style={font=\scriptsize}
    	},
  		mylineblue/.style={
  			line width=1.5pt,
  			every node/.style={font=\scriptsize},
  			color=blue
  		},
  		mylinebluethin/.style={
  			line width=0.5pt,
  			every node/.style={font=\scriptsize},
  			color=blue
  		},
	]
	\input{tikz/bintreeadd.tikz}
\end{tikzpicture}
}
\caption{Subsets of $P = \{p_1,p_2,p_3,p_4\}$.}
\end{subfigure}
\caption{Binomial trees for generating subsets of paths with cardinality up to 3.}
\label{fig:kdsp-bintrees}
\end{figure*}
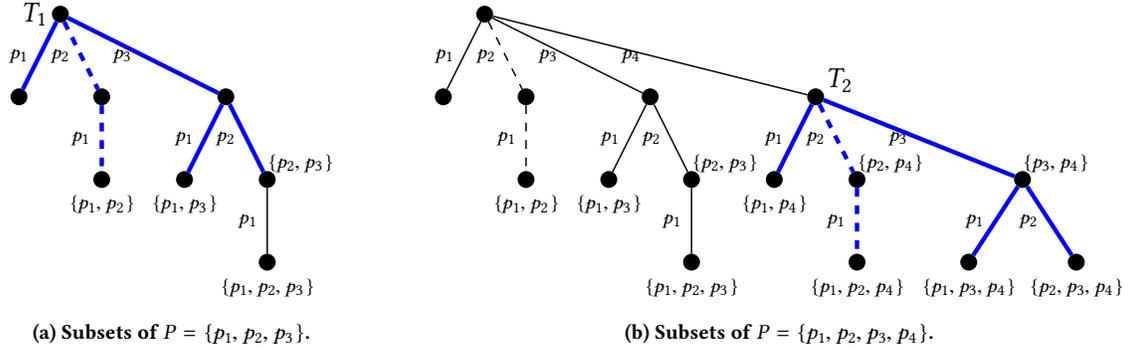

A na{\"i}ve approach for an exact solution to \kdsp would first identify all paths from a source node $s$
to a target node $t$ and examine all possible sets of at most $k$ paths to determine the set that satisfies
the conditions of Problem~1. Such an approach is clearly impractical. In view of this, Section~\ref{sec:kdsp:pruning}
defines a pair of pruning techniques, and Section~\ref{sec:kdsp:exact} presents an exact algorithm
that employs these techniques to reduce the search space during \kdsp query processing.

\subsection{Pruning Techniques} \label{sec:kdsp:pruning}

\paragraph{Lower Bound on Collective Path Length} Our first pruning technique employs a lower bound on the
collective path length, which limits the total number of paths to be constructed. Let \PALL be the set of
all possible paths from node $s$ to $t$, and $p^*_i$ be the $i$-shortest path in \PALL (i.e., $p^*_1 = \sp$).
Then, for every path $p$ in $\PALL$, the collective path length of every
$k$-subset $P \subset \PALL$ that contains $p$ is lower bounded as
follows:
\begin{equation}
\label{eq:lbound}
 \mathcal{L}(P) \geq \ell(p) + \sum_{i=1}^{k{-}1}\ell(p^*_i)
\end{equation}
Intuitively, the right part of the above inequality defines a set of
paths that contain $p$ and the $(k{-}1)$ shortest paths in \PALL; by
definition, such a set should have the lowest collective path length
compared to any subset of \PALL that includes $p$.

We can use the lower bound of Inequality~\ref{eq:lbound} to exclude a path $p$ from the
result set of \PDSP. The idea is captured by the following lemma:
\begin{lemma}
\label{lm:kdsp-term}
Let \PALL be the set of all paths from $s$ to $t$, $P_{A} \subseteq \PALL$ be a set
of $k$ paths that satisfy Condition (A) of Problem~\ref{def:kdsp}, i.e., the contained paths
are sufficiently dissimilar to each other, and $P^*_{k-1}$ be the set of the $k-1$ shortest
paths in \PALL. A path $p \in \PALL \setminus P_{A}$ cannot be part of the \kdsp result if
$\LEN(\{p\}\cup P^*_{k-1}) > \LEN(P_{A})$.
\end{lemma}
\begin{proof}
  By definition, $P^*_{k-1}$ is the set of $k{-}1$
  paths with the minimum possible collective length among all subsets
  containing $k{-}1$ paths. Hence, the subset $P$ which contains $p$
  and achieves the least possible path length sum is
  $P = P^*_{k-1} \cup \{p\}$. Therefore, if there
  exists a set of sufficiently dissimilar paths $P_{A}$, i.e.,
  $\forall p_i,p_j \in P_{A}$, $i \neq j : Sim(p_i,p_j) < \theta$,
  $p \not\in P_{A}$ 
  and $\LEN(P_{A}) < \LEN(P)$, then the path length sum of any
  set $P'$ with $p \in P'$ will be greater than the path length sum of
  $P_{A}$. Consequently, there exist no set of at most $k$ dissimilar
  paths, which contains $p$ (or any path longer than $p$), and
  achieves a collective length smaller that $\LEN(P_{A})$.
\end{proof}

With Lemma~\ref{lm:kdsp-term}, when processing a \kdsp query,
it suffices to construct paths from source node $s$ to target node $t$ in
increasing length order and use the lemma as a
termination condition. If the next path $p$ in length order
cannot be part of the final result set, then the current set of paths 
achieves the lowest possible path length and, hence, it is the final
\PDSP result.

\paragraph{Excluding Subsets of Similar Paths}

Our second pruning technique prevents the construction of path sets
that do not exclusively contain sufficiently dissimilar paths. The key
idea behind this technique is to incrementally generate the path sets
by reusing already generated smaller subsets. For this purpose, we
employ a dynamic programming scheme named "filling a rucksack" or Algorithm $F$ for
simplicity~\cite{knuthTAOCP}, that builds on the concept of \emph{binomial} trees.
Given a set of paths $P$, we use a binomial tree $T$ of height $k$ to represent all
subsets $P_i \subseteq P$ of cardinality $\vert P_i \vert \leq k$. Each time a new
path $p'$ is added to set $P$, Algorithm $F$ extends all existing subsets of
cardinality up to $k{-}1$ to generate the subsets of cardinality up to $k$ that
contain $p'$. More specifically, a new branch is attached to the root of the
binomial tree using an edge labeled by $p'$, and the subtree representing 
subsets of cardinality up to $k{-}1$ is added under the new branch. Note that during
this expansion phase, the height of the binomial tree remains fixed but the tree
becomes wider.

\begin{example}
Consider the example in Figure~\ref{fig:kdsp-bintrees}. Figure~\ref{fig:kdsp-bintrees}a
illustrates the binomial tree of height $k{=}3$ for the path set $P = \{p_1,p_2,p_3\}$. 
The leaf nodes of the tree represent all possible subsets of $P$ containing at most $k{=}3$
paths. Also, the subtree $T_1$ (indicated with  blue lines) represents all possible subsets
of $P$ of cardinality up to $k{-}1=2$. Then, Figure~\ref{fig:kdsp-bintrees}b
illustrates the expanded binomial tree after adding a new path $p_4$ to $P$. 
Observe the new branch attached to the root node via the edge representing $p_4$.
Also, the subtree $T_2$ attached under the new branch (also indicated with blue lines)
is identical to subtree $T_1$.
\end{example}

During the execution of Algorithm F, instead of generating the
contents of each new branch from scratch, the algorithm simply copies the subtree that
represents all subsets of cardinality up to $k{-}1$ before adding the new path.
Consequently, when computing a \kdsp query, it suffices to apply Algorithm $F$ to
incrementally generate new subsets of paths and early prune subsets
that contain at least one pair of not sufficiently dissimilar paths.
In Figure~\ref{fig:kdsp-bintrees} for example, assume that the similarity of paths
$p_1$ and $p_2$ exceeds the given threshold $\theta$, i.e.,
$\olFunc(p_1,p_2) > \theta$. That is, the paths are not sufficiently
dissimilar and consequently, no subset containing paths $p_1$ and
$p_2$ can be part of the \kdsp result. In this case, the branch 
representing the $\{p_1,p_2\}$ subset indicated by the dashed line is excluded
from every subtree of the binomial tree and hence, all subsets that contain both $p_1$
and $p_2$ are never generated.

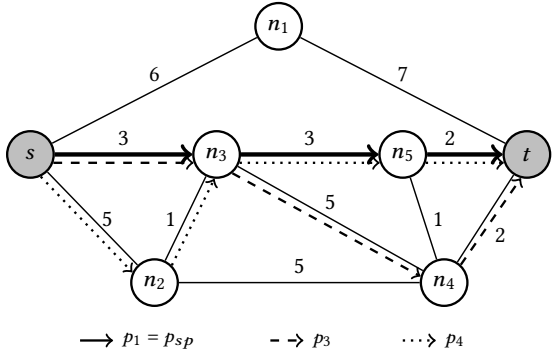
\begin{figure*}[t]
\centering
\scalebox{1.1} {
\begin{tikzpicture}[
		auto,
		node distance=1cm,
 		thick,
		inner/.style={
	  		circle,
  			draw,
  			font=\small,
  			minimum size = 16pt,
  			inner sep=0pt
  		},
  		point/.style={
  			circle,
  			draw,
  			font=\small,
  			fill=lightgray,
  			minimum size = 16pt,
  			inner sep=0pt
  		},
  		box/.style={
  			rectangle,
  			draw,
  			minimum width = 4cm,
  			font=\tiny
  		}
	]
	\input{tikz/bslex-kspd.tikz}
\end{tikzpicture}
}
\caption{\kspdml computing $\kdsp(s,t,3,0.5)$; result paths highlighted in light gray.}
\label{fig:kdsp-bslex}
\end{figure*}

\subsection{The \kspdml Algorithm} \label{sec:kdsp:exact}

We now present our exact \kspdml algorithm for \kdsp, which employs the pruning techniques discussed
in Section~\ref{sec:kdsp:pruning}. The algorithm builds upon the computation of the $K$-Shortest
Paths~\cite{yen1971}, with the goal to progressively compute the exact solution. At each round,
\kspdml employs Algorithm $F$ from \cite{knuthTAOCP} to generate subsets of at most $k$ paths.
Using our second pruning technique, subsets that violate the similarity constraint are filtered out.
The algorithm examines only the subsets with the largest possible cardinality and, among those,
only the one with the lowest collective path length is retained. \kspdml uses our first pruning
technique as the termination condition, i.e., the algorithm terminates after a path that satisfies
the criterion of Lemma~\ref{lm:kdsp-term} is constructed.

Algorithm~\ref{algo:kdsp-algo} illustrates the pseudocode of our exact
\kspdml algorithm, which employs the two aforementioned pruning
techniques. All generated shortest paths are stored in set \PSP. In
Line~2, we set $p$ as the first shortest path from $s$ to $t$. From
Line~3 to~11, \kspdml iterates over the next shortest path starting
from the first one. In Line~4, current path $p$ is stored in \PSP. In
Lines~5--6, the collective length $\mathcal{L}_{k{-}1}$ of the first
$k{-}1$ shortest paths is computed to be used for the lower bound in
Line~3. Next, Algorithm $F$~\cite{knuthTAOCP} is called in Line~7 to
determine all $(\leq k)$-subsets $P$ of \PSP that contain $p$. For
each subset $P$, Line~8 checks whether all paths in $P$ are
sufficiently dissimilar to each other (Condition (A) of
Problem~\ref{def:kdsp}). Subsequently, Line~9 checks whether $P$
contains more paths than $\PDSP$ (Condition (B) of
Problem~\ref{def:kdsp}), or $P$ contains as many paths as \PDSP and has a lower collective length than
current \PDSP (Conditions (B) and (C) of Problem~\ref{def:kdsp}).
If either case holds, \kspdml updates the $\PDSP$ result in Line~10.
The examination of each next shortest
path (and \kspdml overall) terminates when either all possible paths
from $s$ to $t$ have been generated or the termination condition of
Line~3 is met.

\begin{algorithm}[htb]
  \caption{\kspdml}
  \label{algo:kdsp-algo}
  \SetKwFunction{NextShortestPath}{NextShortestPath}%
  \SetKwComment{com}{\hfill $\triangleright$ }{}%
  \KwIn{Road network $G = (N,E)$, source node $s$, target node $t$, number of results $k$, similarity threshold $\theta$}%
  \KwOut{Set $\PDSP$ of at most k paths}%
  \BlankLine%
  \textbf{initialize} 
  $\PSP \gets \emptyset$, 
  $\mathcal{L}_{k{-}1} \gets 0$, $\PDSP \gets \emptyset$, $P \gets \emptyset$\;%
  $p \gets \NextShortestPath(G,s,t)$\com*{Shortest path $\sp$}%
  \BlankLine%
  \While{$p{\neq}\mathit{null}$ \textbf{and} ($\vert \PDSP \vert{<}k$ \textbf{or} $\len(p){+}\mathcal{L}_{k{-}1}{\leq} \LEN(\PDSP)$)}
  {
   	$\PSP \gets \PSP \cup \{p\}$\;%
    \If{$|\PSP| < k$}
    {
    	$\mathcal{L}_{k{-}1} \gets \mathcal{L}_{k{-}1} + \ell(p)$\;%
    }
    \ForEach(\com*[f]{Alg. $F$~\cite{knuthTAOCP}}){$P \subseteq \PSP : \vert P \vert \leq k$ with $p\in P$}{%
      \If(\com*[f]{Problem~\ref{def:kdsp} Condition (A)}){$\forall$ $p_i,p_j\in P$ with $i\neq j : \olFunc(p_i,p_j) < \theta$}
      {
      	\If(\com*[f]{Problem~\ref{def:kdsp} Conditions (B), (C)}){$|P| > |\PDSP|$ \textbf{or} \hspace{5cm} $|P| = |\PDSP|$ \textbf{and} $\LEN(P) < \LEN(\PDSP)$}
      	{
          $\PDSP \gets P$\com*{Update result set}
        }
      }
    }
    $p \gets \NextShortestPath(G,s,t)$\;%
  }%
  \Return $\PDSP$\;%
\end{algorithm}

\begin{example} \label{ex:exact-kdspalgo}
Figure~\ref{fig:kdsp-bslex} illustrates the execution of \kspdml for the $\kdsp(s, t, 3, 0.5)$ query;
on the right-hand side of the figure, we report the $\PSP$ set with the examined paths in length order.
Initially, the shortest path $p_1{=}p_{sp}$ is generated and the result set $\PDSP = \{p_1=\sp\}$ is
initialized. Next, path $p_2$ is computed, but since $\olFunc(p_1,p_2) = 0.545 > 0.5$ subset $\{p_1,p_2\}$
is never constructed. At this point though, $\mathcal{L}_{k{-}1} = 8 + 9 = 17$ is computed. Then, the
algorithm computes the next shortest path in length order, i.e., $p_3$. The computation of $p_3$ results
in no sets of $k$ sufficient dissimilar paths. In fact, there exists a single subset $\{p_1,p_2,p_3\}$
of \PALL that contains $p_3$ but it is never constructed as $\olFunc(p_1,p_2) > \theta$. In contrast,
there exists set $P_1 = \{p_2,p_3\}$ which contains sufficiently dissimilar paths and 
$\vert P_1 \vert > \vert \PDSP \vert$. Hence, \PDSP is updated, i.e., $\PDSP = P_1$. Next, the 4th
shortest path $p_4$ is generated and subsets $P_2 = \{p_1,p_3,p_4\}$ and $P_3 = \{p_2,p_3,p_4\}$ are
constructed both containing more sufficiently dissimilar paths than the current result. Since $P_2$ has
the lowest collective length, i.e., $\LEN(P_2) = 29$, it becomes the temporary result set, i.e. $\PDSP = P_2$.
The algorithm continues its execution as $\len(p_4) + \mathcal{L}_{k{-}1} > \LEN(\{p_1,p_3,p_4\})$
(Lemma~\ref{lm:kdsp-term}).  Note that, since a temporary result of $k$ sufficiently paths has been found,
there is no need to examine subsets of less than $k$ paths from now on. The subsequent generation of paths
$p_5$ and $p_6$ result in the construction of subsets none of which achieve lower collective length
that $P_2$. However, the algorithm terminates only when the 7th shortest path $p_7$ is retrieved, for which
$\len(p_7) + \mathcal{L}_{k{-}1} \geq \LEN(P_2)$ and so, the final result $\PDSP = P_2 = \{p_1,p_3,p_5\}$
is returned.
\end{example}

\subsubsection*{Complexity Analysis}
In the worst case, \kspdml has to examine all $\binom{K}{k}$ subsets containing $k$ $(s{\rightarrow}t)$ paths,
where $K$ is the number of all paths from $s$ to $t$ that the algorithm needs to examine. Even for $k=O(1)$,
the best complexity guarantee one can give for \kspdml is hence $O(\poly(K))$. Since $K$ is not polynomially
bounded in the size of the network $G$, this is prohibitively expensive. In fact, $K$ is usually very large;
for random graphs with density $d$, the expected value is $\mathbb{E}(K)=\Omega((|N|-2)!d^{|N|})$~\cite{roberts2007}.

%% file: tikz/bintree.tikz
\node[point]	(ROOT)	at (0,5)		{};
\node[point]	(N11)	at (-0.5,4)		{};
\node[point]	(N12)	at (0.5,4)		{};
\node[point]	(N13)	at (2,4)		{};
\node[point]	(N21)	at (0.5,3)		{};
\node[point]	(N22)	at (1.5,3)		{};
\node[point]	(N23)	at (2.5,3)		{};
\node[point]	(N31)	at (2.5,2)		{};
\node[font=\scriptsize, color=black]			(T1)	at (0.5,2.7)		{$\{p_1,p_2\}$};
\node[font=\scriptsize, color=black]			(T2)	at (1.5,2.7)		{$\{p_1,p_3\}$};
\node[font=\scriptsize, color=black]			(T3)	at (2.9,3.2)		{$\{p_2,p_3\}$};
\path[myline]
	(N23) 	edge node[left, color=black]		{$p_1$} (N31);
\path[myline,color=blue,line width=1.5pt]
   	(ROOT) edge node[left, color=black]	{$p_1$} (N11)
	(ROOT) edge node[left, color=black]	{$p_3$} (N13)
	(N13) edge node[left, color=black]	{$p_1$} (N22)
	(N13) edge node[left, color=black]	{$p_2$} (N23);
\path[myline,color=blue,line width=1.5pt,dashed]
	(ROOT) edge node[left, color=black]	{$p_2$} (N12)
	(N12) edge node[left, color=black]	{$p_1$} (N21);
\node[font=\scriptsize, color=black]		(L1)	at (2.5,1.7)	{$\{p_1,p_2,p_3\}$};
\node[font=\large,color=black]          (T1)    at (-0.3,5)     {$T_1$};

%% file: tikz/bintreeadd.tikz
\node[fill=white,color=white]   (EMPTY) at (0,5.2)              {};
\node[point]    (ROOT)  at (0,5)                {};
\node[point]    (N11)   at (-0.5,4)             {};
\node[point]    (N12)   at (0.5,4)              {};
\node[point]    (N13)   at (2,4)                {};
\node[point]    (N14)   at (4,4)                {};
\node[point]    (N21)   at (0.5,3)              {};
\node[point]    (N22)   at (1.5,3)              {};
\node[point]    (N23)   at (2.5,3)              {};
\node[point]    (N24)   at (3.5,3)              {};
\node[point]    (N25)   at (4.5,3)              {};
\node[point]    (N26)   at (6.5,3)              {};
\node[point]    (N31)   at (2.5,2)              {};
\node[point]    (N32)   at (4.5,2)              {};
\node[point]    (N33)   at (5.85,2)             {};
\node[point]    (N34)   at (7.15,2)             {};
\node[font=\scriptsize, color=black]			(T1)	at (0.5,2.7)		{$\{p_1,p_2\}$};
\node[font=\scriptsize, color=black]			(T2)	at (1.5,2.7)		{$\{p_1,p_3\}$};
\node[font=\scriptsize, color=black]			(T3)	at (2.9,3.2)		{$\{p_2,p_3\}$};
\node[font=\scriptsize, color=black]			(T1)	at (3.5,2.7)		{$\{p_1,p_4\}$};
\node[font=\scriptsize, color=black]			(T2)	at (4.9,3.2)		{$\{p_2,p_4\}$};
\node[font=\scriptsize, color=black]			(T2)	at (6.9,3.2)		{$\{p_3,p_4\}$};
\node[font=\scriptsize, color=black]            (T1)    at (2.5,1.7)    {$\{p_1,p_2,p_3\}$};
\node[font=\scriptsize, color=black]            (T2)    at (4.5,1.7)    {$\{p_1,p_2,p_4\}$};
\node[font=\scriptsize, color=black]            (T3)    at (5.85,1.7)           {$\{p_1,p_3,p_4\}$};
\node[font=\scriptsize, color=black]            (T4)    at (7.15,1.7)           {$\{p_2,p_3,p_4\}$};
\path[myline]
        (ROOT)  edge node[left,color=black]         {$p_4$} (N14)
        (N23)   edge node[left,color=black]         {$p_1$} (N31)
        (ROOT)  edge node[left,color=black]         {$p_1$} (N11)
        (ROOT)  edge node[left,color=black]         {$p_3$} (N13)
        (N13)   edge node[left,color=black]         {$p_1$} (N22)
        (N13)   edge node[left,color=black]         {$p_2$} (N23);
\path[mylineblue]
        (N14)   edge node[left,color=black]         {$p_1$} (N24)
        (N14)   edge node[left,color=black]         {$p_3$} (N26)
        (N26)   edge node[left,color=black]         {$p_1$} (N33)
        (N26)   edge node[left,color=black]         {$p_2$} (N34);
\path[mylineblue,dashed]
        (N14)   edge node[left,color=black]             {$p_2$} (N25)
        (N25)   edge node[left,color=black]             {$p_1$} (N32);
\path[mylineredthin]
        (ROOT)  edge node[left,color=black]             {$p_2$} (N12)
        (N12)   edge node[left,color=black]             {$p_1$} (N21);
%
\node[font=\large,color=black]          (T1)    at (4.3,4.2)    {$T_2$};

%% file: tikz/bslex-kspd.tikz
\input{tikz/basegraph.tikz}

	\draw[->] 			(-2.4,-0.5) 	-- 	(-2,-0.5) node[anchor=west]{\scriptsize $p_1=\sp$};
	\draw[->,dashed] 	(-0.1,-0.5) -- 	(0.3,-0.5) node[anchor=west]{\scriptsize $p_3$};	
	\draw[->,dotted] 	(1.5,-0.5) -- 	(1.9,-0.5) node[anchor=west]{\scriptsize $p_4$};
%
%
\path[->, line width=1.5pt, every node/.style={font=\scriptsize}]
    (N0) edge node[above]       {} (N3)
    (N3) edge node[above]       {} (N5)
    (N5) edge node[above]       {} (N6);
\path[->, dashed, every node/.style={font=\scriptsize}]
    (N0.340) edge node[above]       {} (N3.200)
    (N3.310) edge node[above]       {} (N4.165)
    (N4.47) edge node[above]        {} (N6.256);
\path[->, dotted, every node/.style={font=\scriptsize}]
    (N0.295) edge node[above]       {} (N2.155)
    (N2.45) edge node[above]       {} (N3.265)
    (N3.340) edge node[above]        {} (N5.200)
    (N5.340) edge node[above]        {} (N6.200);
\node (T)       at (7.5,1.8)    {
    \resizebox{0.85\columnwidth}{!}{%
    \begin{tabular}{|c|c|c|}
    	\hline
        \multicolumn{2}{|c|}{\textbf{$\PALL$}} & \textbf{$length$} \\
        \hline\hline
        \cellcolor{gray!20}$p_1 = \sp$ &	\cellcolor{gray!20}$\langle (s,n_3),(n_3,n_5),(n_5,t) \rangle$  			&   \cellcolor{gray!20}8  \\ 
        \hline
        $p_2$   	&   $\langle (s,n_3),(n_3,n_5),(n_5,n_4),(n_4,t) \rangle$   &   9  \\ 
        \hline
        \cellcolor{gray!20}$p_3$    &   \cellcolor{gray!20}$\langle (s,n_3),(n_3,n_4),(n_4,t) \rangle$  			&  \cellcolor{gray!20}10 \\ 
        \hline
        \cellcolor{gray!20} $p_4$   &   \cellcolor{gray!20}$\langle (s,n_2),(n_2,n_3),(n_3,n_5),(n_5,t) \rangle$   &    \cellcolor{gray!20}11 \\ 
        \hline
        $p_5$   	&  $\langle (s,n_3),(n_3,n_2),(n_2,n_4),(n_4,t) \rangle$   &   11 \\ 
        \hline
        $p_6$       &   $\langle (s,n_3),(n_3,n_4),(n_4,n_5),(n_5,t) \rangle$ 	&   11 \\ 
        \hline
        $p_7$   	&   $\langle (s,n_2),(n_2,n_3),(n_3,n_5),(n_5,n_4),(n_4,t) \rangle$   &   12 \\ 
        \hline
    \end{tabular}
    }
};


%% file: chapters/heuristic.tex
\section{Heuristic Approaches} \label{sec:heuristic} 

Our complexity analysis in Section~\ref{sec:kdsp:exact} showed that the cost of \kspdml
is prohibitively high due to the number $K$ of paths that need to be examined;
in practice, we expect this $K$ to be significantly larger than the requested number of
results $k$. Consequently, solutions that build on the computation of the $K$-shortest paths,
i.e., both our exact \kspdml and the heuristic \kspdplus algorithm of Liu et al.~\cite{liu2017}, cannot
scale to real-world road networks. In this spirit, Section~\ref{sec:heuristic-ssvp} introduces
the concept of \emph{simple single-via} paths which allow us to drastically reduce the search
space. Then, Sections~\ref{sec:heuristic-svpdml} and~\ref{sec:heuristic-svpdplus} present our
heuristic algorithms for the \kdsp problem that iterate over the set of simple single-via paths
in length order.

\subsection{Simple Single-Via Paths}\label{sec:heuristic-ssvp}

The concept of \emph{single-via paths} (SVP) initially proposed by Abraham
et al.~\cite{abraham2013} and then optimized by Luxen and Schieferdecker~\cite{luxen2015},
has been primarily used for alternative routing on road networks \cite{abraham2013,luxen2015,chondrogiannis2017}.
Given a road network $G{=}(N,E)$, a source node $s$ and a target node $t$, the single-via path
$p_{sv}(n)$ of a node $n \in N \setminus \{s,t\}$ is defined as $\sp(s{\rightarrow}n) \circ \sp(n{\rightarrow}t)$,
i.e., the concatenation of shortest paths $\sp(s{\rightarrow}n)$ and $\sp(n{\rightarrow}t)$.
By definition, $p_{sv}(n)$ is the shortest possible path that connects $s$ and $t$ through $n$.

However, using the set of single-via paths to process \kdsp queries
raises two important issues.  First, the single-via path $p_{sv}(n)$
of every node $n$ crossed by the shortest path $\sp(s{\rightarrow}t)$
is identical to \sp. Computing the single-via paths for these
particular nodes is unnecessary.  Second, there is no guarantee that a
single-via path is simple (i.e., cycle-free), which may result in
recommending paths that make little sense from a user perspective. For
instance, consider the road network in Figure~\ref{fig:kdspsvp}. The
single-via path of node $n_2$ is
$p_{sv}(n_2){=}\sp(s{\rightarrow}n_2) \circ \sp(n_2{\rightarrow}t)$,
which is clearly not a simple path, i.e.,
$p_{sv}(n_2){=}\langle (s,n_3),(n_3,n_2),(n_2,n_3),(n_3,n_5),(n_5,t)
\rangle$.

To address the aforementioned issues, we introduce the \emph{simple
single-via paths} (SSVP). Given a road network $G{=}(N,E)$, a source node
$s$ and a target node $t$, the SSVP $p_{ssv}(n)$ of a node
$n \in N \setminus \{s,t\}$ is defined only if $n$ does not lie on the
shortest path $\sp(s{\rightarrow}t)$. If the single-via path
$p_{sv}(n)$ of node $n$ is simple, then $p_{ssv}(n){=}p_{sv}(n)$.
Otherwise, $p_{ssv}(n)$ is the concatenation either of
$\sp(s{\rightarrow}n)$ with the shortest path $p'(n{\rightarrow}t)$
from $n$ to $t$ that visits no nodes in $\sp(s{\rightarrow}n)$, i.e.,
$sp(s{\rightarrow}n) \circ p'(n{\rightarrow}t)$, or the concatenation
of the shortest path $p'(s{\rightarrow}n)$ from $s$ to $n$ that visits
no nodes in $\sp(n{\rightarrow}t)$ with $\sp(n{\rightarrow}t)$, i.e.,
$p'(s{\rightarrow}n) \circ \sp(n{\rightarrow}t)$. In this case, the
SSVP of $n$ is the concatenated path with the lowest path length.
Note that the shortest path $\sp(s{\rightarrow}t)$ is a SSVP by
definition.

Consider again the road network in Figure~\ref{fig:altrouting-alternativeex}. As mentioned already,
the single-via path $p_{sv}(n_2)$ is not simple. Hence, the simple single-via path $p_{ssv}(n_2)$
is either $p = \langle (s,n_2), (n_2,n_3), (n_3,n_5), (n_5,t) \rangle$
or  $p' = \langle (s,n_n3), (n_3,n_2), (n_2,n_4), (n_4,t) \rangle$. In this particular case,
both concatenated paths have the same length, i.e., $\len(p) = \len(p') = 11$, and either can
be set as $p_{ssv}(n_2)$. The table on the right-hand side of Figure~\ref{fig:kdspsvp} illustrates
the simple single-via paths for the road network of our running example.

\subsection{The \svpdml Algorithm} \label{sec:heuristic-svpdml}

A straightforward way of employing SSVPs for processing \kdsp queries is to alter the exact \kspdml algorithm
from Section~\ref{sec:kdsp:exact} such that the algorithm iterates over the simple single-via paths
from source $s$ to target $t$ in increasing length order, instead of retrieving the next shortest path.
By the definition of SSVP, the search space will be drastically reduced as at most $\vert N \vert{-}1$ paths
will be examined. In addition, we can still use the pruning techniques proposed in Section~\ref{sec:kdsp:pruning}
to terminate the search and avoid generating candidate path sets that do not contain sufficiently dissimilar
paths. Our first heuristic algorithm, termed \svpdml, follows this approach.
Naturally, as the result of \kdsp queries may not consist exclusively of simple single-via paths,
\svpdml can only provide an approximate solution to the \kdsp problem.

\begin{algorithm}[t]
\small
	\caption{\svpdml}
	\label{algo:kdsp-svp}
	\SetKwFunction{ConstructSSVP}{ConstructSimpleSingleViaPaths}%
	\SetKwFunction{NextSSVP}{NextSSVPByLength}%
    \SetKwComment{com}{\hfill $\triangleright$ }{}%
	\KwIn{Road network $G = (N,E)$, source node $s$, target node $t$, $\#$ of results $k$, similarity threshold $\theta$}%
	\KwOut{Set $\PDSP$ of at most $k$ paths}%
	\BlankLine%
	\textbf{initialize} $\PSSV \gets \emptyset$, $\mathcal{L}_{k{-}1} \gets 0$,$\PDSP \gets \emptyset$,$P \gets \emptyset$\;%
  	$p \gets \NextSSVP(G, s, t)$\com*{Shortest path $\sp$}%
  	\While(\com*[f]{Lemma~\ref{lm:kdsp-term}}){$p \neq \mathit{null}$ \textbf{and} \hspace{6cm} ($\vert \PDSP \vert{<}k$ 		
  	\textbf{or} $\len(p){+}\mathcal{L}_{k{-}1}{\leq} \LEN(\PDSP)$)}
  	{
    	$\PSSV \gets \PSSV \cup \{p\}$\;%
    	\If{$|\PSSV| < k$}
    	{
    		$\mathcal{L}_{k{-}1} \gets \mathcal{L}_{k{-}1} + \ell(p)$\;%
    	}	
    	\ForEach(\com*[f]{Alg. $F$~\cite{knuthTAOCP}}){$P \subseteq \PSSV : \vert P \vert \leq k$ with $p\in P$}{%
      	\If{$\forall$ $p_i,p_j\in P$ with $i\neq j : \olFunc(p_i,p_j) < \theta$}
      	{
      		\If{$|P| > |\PDSP|$ \textbf{or} \hspace{5cm} $|P| = |\PDSP|$ \textbf{and} $\LEN(P) < \LEN(\PDSP)$}
      		{
         		 $\PDSP \gets P$\com*{Update result set}
        	}
      	}
    }
    $p \gets \NextSSVP(G, s, t)$\;
  }
  \Return $\PDSP$\;
\end{algorithm}

Algorithm~\ref{algo:kdsp-svp} illustrates the pseudocode of \svpdml. The algorithm keeps track
of all generated simple single-via paths inside set $\PSSV$ (instead of the shortest paths
inside \PSP, in Algorithm~\ref{algo:kdsp-algo}). \svpdml proceeds exactly as \kspdml but replaces
the call to \NextShortestPath by a call to the \NextSSVP function, which retrieves the
next simple single-via path $p(s{\rightarrow}t)$ in increasing length order. 
In Line~2, the first simple single-via path, i.e., the $p_{sp}(s \rightarrow t)$ shortest path, is retrieved.
From Line~3 to 11, \svpdml iterates over each next simple single-via path $p$ in order, to determine the
subsets of $\PSSV$ that contain $p$ and satisfy the conditions of Problem~\ref{def:kdsp} (Lines~8--9). Upon
identifying such subsets the algorithm updates $\PDSP$ if needed, in Line~10, similar to \kspdml. 
Also like \kspdml, \svpdml terminates after all simple single-via paths have been examined or if
current single simple-via path $p$ satisfies the termination condition of Lemma~\ref{lm:kdsp-term}. 

Next, we elaborate on \NextSSVP. We design the function as an iterator which allows us to compute 
the simple single-via paths in increasing length order.
Function~\ref{algo:nextssvp} illustrates the pseudocode of \NextSSVP. Upon the first call (Lines 1--11),
the function constructs the shortest path trees $T_{s{\rightarrow}N}$ and $T_{N{\rightarrow}t}$ and the shortest
path $\sp$ from $s$ to $t$; note that $T_{N{\rightarrow}t}$ is computed by reversing the direction
of the network edges. In addition, all nodes $n$ of the road network that are not crossed by $\sp$ are
organized inside min-priority queue \Q according to the length of their single-via path,
i.e., $\len(p_{s{\rightarrow}n} \circ p_{n{\rightarrow}t})$.
Note that \Q is static, i.e., its contents are preserved throughout every call of \NextSSVP.
Last, the shortest path $\sp$ is returned as the first simple single-via path in Line~11.
Every followup function call to compute the next simple single-via path in length order, is handled by Lines~12--24.
In Line~13 the node at the top of \Q is extracted and, in Line~14, \NextSSVP checks
whether single-via path $p_{sv}(n) = p_{s{\rightarrow}n} \circ p_{n{\rightarrow}t}$ is simple. If so,
then the path is returned. By definition a simple single-via path of a node $v$ is at least as long
as the single-via path of $v$, i.e., $\ell(p_{ssv}(n)) \geq \ell(p_{sv}(n))$. Thus, as the remaining
nodes in \Q have a longer single-via path than top node $n$, the returned path is indeed the next
simple single-via path in increasing length order. 

If path $p_{s{\rightarrow}n} \circ p_{n{\rightarrow}t}$ is not simple,
then function \NextSSVP needs to construct the simple single-via path of current node $n$.
This construction involves three steps in Lines~17--26. First, the subgraphs $G_1$ and $G_2$ of the
original network are defined by excluding nodes that are crossed by shortest paths $p_{s{\rightarrow}n}$
and $p_{n{\rightarrow}t}$, respectively, in Lines~17--20.
Then, in Lines~21--22, the \NextSSVP computes paths $p_1$ and $p_2$, which
are the candidate paths for being the simple single-via path of $n$.
Note that the shortest paths that are computed on $G_1$ and $G_2$ instead of the original
road network $G$ guarantee that the candidate paths $p_1$ and $p_2$ are simple.
Last, in Lines~23--26, the shortest path among $p_1$ and $p_2$ is selected as the simple single-via
path of current node $n$, and is subsequently inserted to \Q. 
This, step is necessary to guarantee that the function always returns the next simple single-via path in length order. 
Every call to \NextSSVP returns the next simple single-via paths in length order until \Q is empty, i.e., until
all simple single-via paths have been generated.

\setcounter{algocf}{0}
\begin{myfunction}[t]
\small
	\caption{NextSSVPByLength}
	\label{algo:nextssvp}
	\SetKwFunction{RetrieveSP}{RetrieveShortestPath}%
	\SetKwFunction{GetSP}{ComputeShortestPath}%
    \SetKwComment{com}{\hfill $\triangleright$ }{}%
	\KwIn{Road network $G(N,E)$, source node $s$, target node $t$}%
	\KwOut{Next simple single-via path by length}%
	\BlankLine%
	\If {first call} {
		\textbf{initialize} static min-priority queue $\Q \gets \emptyset$\;%
		$T_{s{\rightarrow}N} \gets$ shortest path tree from $s$ to all $n \in N$\;
		$T_{N{\rightarrow}t}\gets$ shortest path tree from all $n \in N$ to $t$\;
		$\sp \gets \GetSP(G,s,t)$\;
		\ForEach {$n \in N - \{s,t\}$} {
			\If{$n$ is not crossed by $\sp$} {
				$p_{s \rightarrow n} \gets \RetrieveSP(T_{s{\rightarrow}N},s,n)$\;%
				$p_{n \rightarrow t} \gets \RetrieveSP(T_{N{\rightarrow}t},n,t)$\;%
    			$\Q.push(\langle n,\len(p_{s \rightarrow n}){+}\len(p_{n \rightarrow t}), p_{s \rightarrow n} \circ p_{n \rightarrow t}  \rangle)$\;%
    		}
		}
		\Return $\sp$\;
	}
	\BlankLine%
	\While{$\Q$ is not empty} {	
		$\langle n,\len(p_{s \rightarrow n} \circ p_{n \rightarrow t}), p_{s \rightarrow n} \circ p_{n \rightarrow t} \rangle \gets \Q.pop()$\;%
		\uIf{$p_{s \rightarrow n} \circ p_{n \rightarrow t}$ is simple} { 
			\Return $p_{s \rightarrow n} \circ p_{n \rightarrow t}$\;	
		}
		\Else {	
			\textbf{let} $N_1 = N \setminus \{$nodes crossed by $p_{s \rightarrow n}\}$ except $s,n$\;
			\textbf{let} $N_2 = N \setminus \{$nodes crossed by $p_{n \rightarrow t}\}$ except $n,t$\;
			\textbf{let} $E_1 = E \setminus p_{s \rightarrow n}$\;
			\textbf{let} $E_2 = E \setminus p_{n \rightarrow t}$\;
        	$p_1 \gets p_{s \rightarrow n} \circ \GetSP(G_1(N_1,E_1),n,t)$\;%
        	$p_2 \gets \GetSP(G_2(N_2,E_2),s,n) \circ  p_{n \rightarrow t}$\;%
			\uIf{$\len(p_1) \leq \len(p_2)$} {
				$\Q.push(\langle n,\len(p_1),p_1 \rangle)$\;%
        	}
        	\Else {
    			$\Q.push(\langle n,\len(p_2),p_2 \rangle)$\;%
    		}
 		}	
	}
  	\Return null\;	
\end{myfunction}

\begin{figure*}[htb]
\centering
\scalebox{1.1} {
\begin{tikzpicture}[
		auto,
		node distance=1cm,
  		thick,
  		inner/.style={
	  		circle,
  			draw,
  			font=\small,
  			minimum size = 16pt,
  			inner sep=0pt
  		},
  		point/.style={
  			circle,
  			draw,
  			font=\small,
  			fill=lightgray,
  			minimum size = 16pt,
  			inner sep=0pt
  		},
  		box/.style={
  			rectangle,
  			draw,
  			minimum width = 4cm,
  			font=\tiny
  		}
	]
	\input{tikz/kdspsvp.tikz}
\end{tikzpicture}
}
\caption{\svpdml and \svpdplus computing $\kdsp(s,t,3,0.5)$; result paths highlighted in light gray.}
\label{fig:kdspsvp}
\end{figure*}
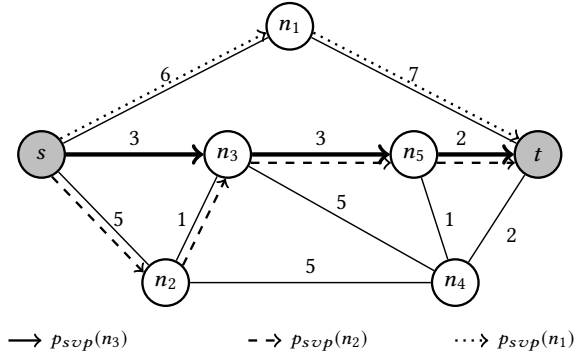

\begin{example}\label{ex:svpdml}
Figure~\ref{fig:kdspsvp} illustrates the execution of \svpdml for the $\kdsp(s, t, 3, 0.5)$ query. On the
right-hand side of the figure we report the set of simple single-via paths $\PSSV$ examined by the
algorithm in increasing length order. Initially, the first simple single-via path, i.e., the shortest
path $\sp$, is retrieved and the result set $\PDSP = \{\sp\}$ is initialized. Next, simple single-via
path $p_{ssv}(n_4)$ is retrieved. Since $\olFunc(\sp,p_{ssv}(n_4)){=}0.545 > 0.5$ subset
$\{\sp,p_{ssv}(n_4)\}$ is never constructed. At this point, $\mathcal{L}_{k{-}1}{=}8{+}9{=}17$ is
computed. Then, the next single-via path in length order $p_{ssv}(n_2)$ is retrieved. The retrieval of
$p_{ssv}(n_2)$ results in no sets of $k$ sufficient dissimilar as the only subset of three paths
$\{\sp, p_{ssv}(n_4), p_{ssv}(n_2)\}$ is never constructed. Among the subsets containing two paths,
set $P_1 = \{\sp, p_{ssv}(n_2)\}$ has the the lowest collective length, i.e., $\LEN(P_1) = 19$. Hence, at
this point \PDSP is set equal to $P_1$. Subsequently, the next simple single-via path in length
order $p_{ssv}(n_1)$ is retrieved. The retrieval of $p_{ssv}(n_1)$ results in the creation of subsets
$P_2{=}\{\sp, p_{ssv}(n_2),p_{ssv}(n_1)\}$ and $P_3{=}\{p_{ssv}(n_4), p_{ssv}(n_2),p_{ssv}(n_1)\}$.
Consequently, as $\LEN(P_2) < \LEN(P_3)$ the result set \PDSP is updated, i.e., $\PDSP = P_2$. Since
there are no more simple single-via paths left to examine, the algorithm terminates.
\end{example}

\subsection{The \svpdplus Algorithm} \label{sec:heuristic-svpdplus}

Despite reducing the search space compared to \kspdml, \svpdml still needs to examine all possible subsets
of at most $k$ sufficiently dissimilar paths. The generation of these subsets is accelerated by the use of
the binomial trees (as discussed in Section~\ref{sec:kdsp:pruning}), but overall we expect this procedure to
dominate the evaluation of the \kdsp query, rendering \kspdml inefficient for real-world networks. In view of this,
we devise a second heuristic algorithm, termed \svpdplus, which adopts a similar idea to
\kspdplus and to our previous work~\cite{chondrogiannis2015,chondrogiannis2017}. The algorithm 
constructs progressively an approximate solution to \kdsp that (1)
always contains the shortest path from source node $s$ to target $t$ and (2) every newly added path to the result
set is both sufficiently dissimilar to the previously recommended paths and as short as possible.
As discussed in Section~\ref{sec:related}, \kspdplus builds on top of the computation of the
$K$-shortest paths while employing  lower bounds to postpone the examination of non-promising paths.  
On the contrary, \svpdplus builds on top of the computation of the simple single-via paths.
Even though \svpdplus cannot employ the same pruning criteria with \kspdplus, the use of simple
single-via paths significantly limits its search space.

Algorithm~\ref{algo:svpdplus} illustrates the pseudocode of \svpdplus.
Similar to \svpdml, \svpdplus invokes \NextSSVP to retrieve the next simple single-via path in length order.
In the beginning, the algorithm retrieves the $\sp(s{\rightarrow}t)$ shortest path,
in Line~2.
From Line~3 to 6, \svpdplus iterates over each next simple single-via path $p$.
In Line~4 the similarity of current path $p$ with all the paths already in \PDSP is checked.
If $p$ is sufficiently dissimilar to all paths in \PDSP, then it is added to \PDSP in
Line~5. Then, the next simple single-via path in length order is retrieved in Line~6.
The algorithm continues its execution until either \PDSP contains $k$ paths, or
there are no more simple single-via paths left to examine.

\setcounter{algocf}{2}
\begin{algorithm}[t]
\small
	\caption{\svpdplus}
	\label{algo:svpdplus}
	\SetKwComment{com}{\hfill $\triangleright$ }{}
	\KwIn{Road network $G(N,E)$, source node $s$, target node $t$, $\#$ of results $k$, similarity threshold $\theta$}%
	\KwOut{Set $\PDSP$ of at most $k$ paths}%
	\BlankLine%
  	$\PDSP \gets \emptyset$\;
	$p \gets \NextSSVP(G,s,t)$\com*{Shortest path $\sp$}%
	\While {\PDSP contains less than $k$ paths \textbf{and} $p$ is not null} {
		\If{$\olFunc(p,p') < \theta$ for all $p' \in \PDSP$}
		{
			add $p$ to $\PDSP$\com*{Update result set}
		}
		$p \gets \NextSSVP(G,s,t)$\;
	}
	\Return $\PDSP$\;
\end{algorithm}

\begin{example}\label{ex:svpdplus}
We illustrate the execution of \svpdplus for the $\kdsp(s, t, 3, 0.5)$ query using again Figure~\ref{fig:kdspsvp}; on the right-hand
side, we report the set of simple single-via paths $\PSSV$ examined by the algorithm in increasing length order.
Initially, the first simple single-via path, i.e., the shortest path $\sp$, is retrieved and directly added to the 
result set \PDSP. Next, simple single-via path $p_{ssv}(n_4)$ is retrieved, but since
$\olFunc(\sp,p_{ssv}(n_4)) = 0.545 > 0.5$, it is not added to \PDSP. The next simple single-via path $p_{ssv}(n_2)$ is then
retrieved. Since, $\olFunc(\sp,p_{ssv}(n_2)) = 0.36 < \theta$, it is added to the result set, i.e., $\PDSP = \{\sp,p_{ssv}(n_2)\}$.
Finally, the simple single-via path $p_{ssv}(n_1)$ is retrieved. Since $p_{ssv}(n_1)$ is sufficiently dissimilar to 
all paths in \PDSP ($\olFunc(\sp,p_{ssv}(n_1)) = 0 < \theta$ and $\olFunc(p_{ssv}(n_2),p_{ssv}(n_1)) = 0 < \theta$), it 
is added to the result set. At this point, $\PDSP = \{\sp,p_{ssv}(n_2),p_{ssv}(n_1)\}$ contains exactly $k$ paths and so,
the algorithm terminates.
\end{example}

%% file: tikz/kdspsvp.tikz
\input{tikz/basegraph.tikz}
%
%
%
	\draw[->] 			(-3.4,-0.5) 	-- 	(-3,-0.5) node[anchor=west]{\scriptsize $p_{svp}(n_3)$};
	\draw[->,dashed] 	(-0.5,-0.5) -- 	(-0.1,-0.5) node[anchor=west]{\scriptsize $p_{svp}(n_2)$};	
	\draw[->,dotted] 	(2,-0.5) -- 	(2.4,-0.5) node[anchor=west]{\scriptsize $p_{svp}(n_1)$};
\path[->, line width=1.5pt, every node/.style={font=\scriptsize}]
    (N0) edge node[above]       {} (N3)
    (N3) edge node[above]       {} (N5)
    (N5) edge node[above]       {} (N6);
\path[->, dashed, every node/.style={font=\scriptsize}]
        (N0.295) edge node[above]       {} (N2.155)
        (N2.45) edge node[above]       {} (N3.265)
        (N3.340) edge node[above]        {} (N5.200)
        (N5.340) edge node[above]        {} (N6.200);
\path[->, dotted, every node/.style={font=\scriptsize}]
        (N0.40) edge node[above]        {} (N1.195)
        (N1.345) edge node[above]       {} (N6.140);
\node (T)       at (7,1.8)      {
        \resizebox{0.75\columnwidth}{!}{%
        \begin{tabular}{|c|c|c|}
                \hline
                \multicolumn{2}{|c|}{\emph{\PSSV}} & \emph{length}\\
                \hline\hline
                $\cellcolor{gray!20}\sp$      &       $\cellcolor{gray!20}\langle (s,n_3),(n_3,n_5),(n_5,t) \rangle$  &       \cellcolor{gray!20}8   \\
                \hline
                $p_{svp}(n_4)$                                          &   $\langle (s,n_3),(n_3,n_5), (n_5,n_4),(n_4,t) \rangle$              &   9                                   \\
                \hline
                $\cellcolor{gray!20}p_{svp}(n_2)$       &   \cellcolor{gray!20}$\langle (s,n_2),(n_2,n_3),(n_3,n_5),(n_5,t) \rangle$  &   \cellcolor{gray!20}11   \\
                \hline
                $\cellcolor{gray!20}p_{svp}(n_1)$       &   \cellcolor{gray!20}$\langle (s,n_1),(n_1,t) \rangle$      &   \cellcolor{gray!20}13  \\
                \hline
        \end{tabular}
        }
};

%% file: chapters/experiments.tex
\section{Experimental Evaluation} \label{sec:exps}

Our experimental analysis involves four publicly available real-world road
networks, i.e., the road network of the city of Adlershof extracted from
OpenStreetMap\footnote{\url{https://www.openstreetmap.org/}},
and the road networks of the cities of Surat, Tianjin
and Beijing~\cite{karduni2016}. Table~\ref{tab:datasets} shows the size
of the aforementioned road networks. Apart from
our exact algorithm \kspdml and our two heuristic algorithms \svpdml
and \svpdplus, we also include in the evaluation the \kspdplus
algorithm, i.e., the greedy heuristic approach proposed by Liu et
al.~\cite{liu2017}.  All algorithms were implemented in C++, and the
tests run on a machine with 2 Intel Xeon E5-2667 v3 (3.20GHz)
processors and 96GB RAM running Ubuntu Linux.

To assess the performance of our
algorithms, we measure the average response time over 1,000 random
queries (i.e., pairs of nodes), varying the number $k$ of requested
paths and the similarity threshold $\theta$. In each experiment, we
vary one of the two parameters and fix the other to its default value,
i.e., $3$ for $k$ and $0.5$ for $\theta$. Apart from the runtime, we
also report on the quality of the results of each algorithm and the
completeness of each result set.  For the quality, we measure the
average length of the computed paths and compare it to the length of
the shortest path. For the completeness, we measure the percentage of
queries for which an algorithm returns exactly $k$ paths. 

\input{tables/datasets}

\subsection{Runtime}
Figures~\ref{fig:kdsp-exp-k} and~\ref{fig:kdsp-exp-theta} report on
the response time of the algorithms varying the requested number of
paths $k$ and the similarity threshold $\theta$, respectively. First,
we observe that the exact algorithm \kspdml is clearly impractical.
For the road network of Adlershof (Figure~\ref{fig:kdsp-exp-k}a) and
the default values of $k$ and $\theta$, \kspdml may provide reasonable
response time, but it is at least one order of magnitude slower than
its competitors. Furthermore, \kspdml requires more than 3 seconds on
average to process \kdsp queries for $k{>}3$ or $\theta{<}0.5$. The
algorithm provides reasonable response times only for $k{=}2 $ and
$\theta{>}0.5$. Since the response time of \kspdml is prohibitively
high even on a road network as small as Adlershof, it is clear that
the algorithm cannot be used on larger road networks; hence, we
exclude \kspdml from the experiments on larger road networks.

\input{charts/kdsp-exp-k}

Next, we elaborate on the performance of the heuristic algorithms
w.r.t. to the number of requested paths $k$. First, we observe in
Figure~\ref{fig:kdsp-exp-k} that \svpdplus is the fastest algorithm in
all cases, beating all its competitors by almost one order of
magnitude for $k{=}2$ and by approximately two orders of magnitude for
$k{>}2$. For \kspdplus, we observe that the runtime of the algorithm
increases with $k$, but this increase is not very steep.  In fact,
while in all networks {\kspdplus} requires much more time than the other
heuristic algorithms to compute a single path dissimilar to the
shortest path, i.e., for $k{=}2$, computing every followup path is not
as expensive. On the other hand, the increase in the runtime of
\svpdml is much more abrupt.  In particular, for $k{=}2$, \svpdml is
faster than \kspdplus in all road networks, and for $k{=}3$, \svpdml is
slower on in Tianjin. However, for $k{>}3$ \svpdml becomes at least
one order of magnitude slower that \kspdplus.

\input{charts/kdsp-exp-theta}

With regard to the similarity threshold $\theta$, in
Figure~\ref{fig:kdsp-exp-theta} we observe that \svpdplus is again the
fastest algorithm in all cases, with improvements of more than two
orders of magnitudes in most cases.
For \kspdplus we observe a similar behavior as before. While the runtime
of \kspdplus increases with a decreasing $\theta$, the increase is not
abrupt. On the contrary, the runtime of \svpdml increases much more with a decreasing
$\theta$. Nevertheless, for \kspdplus and \svpdml, we observe that
\svpdml is faster for $\theta{>}0.5$ with the exception of Tianjin,
while \kspdplus is always faster for $\theta{<}0.5$.

\subsection{Quality and Completeness}

In Figure~\ref{fig:kdsp-exp-quality}, we present our findings on the
quality of the computed results. We consider all queries for which
each algorithm returned exactly $k$ paths and compute the average
length of the returned paths. Then we compare the average length of
each result set to the length of the shortest path. That is, we show
how much longer, on average, the alternative paths with respect to the
shortest path are. Apparently, as shown in Figure~\ref{fig:kdsp-exp-quality}a,
the exact \kspdml algorithm computes the shortest alternative paths on
average. Looking at the heuristic solutions, \kspdplus produces paths
with an average length that is very close to the exact solution.
\svpdml comes next, while \svpdplus recommends the paths with the
highest length on average. However, we observe in Figures~\ref{fig:kdsp-exp-quality}a-c
that the difference between the result set of \kspdplus and the result sets
of \svpdml and \svpdplus is, in all cases, almost insignificant. In particular,
for the road networks of Surat, Tianjin and Beijing, the difference is always less
than $2\%$ for \svpdml and less than $6\%$ for \svpdplus.

\input{charts/kdsp-exp-quality}

Finally, we report on the completeness of the result set of each algorithm.
Table~\ref{tab:completeness} reports for each algorithm the percentage of queries for
which exactly $k$ paths were found. Naturally, the exact algorithm \kspdml demonstrates the highest
completeness ratio. However, all heuristic algorithms are very close to the exact solution.
In most cases, all algorithms demonstrate a completeness ratio of more
than $94\%$. The only case where there is a notable difference to the completeness ratio of the algorithms is 
for $\theta{=}0.1$. In this case, we observe that the completeness ratio of \kspdplus
and \svpdplus is significantly lower than the ratio of \kspdml and \svpdml. For the road network of Adlershof 
in particular, the completeness ratio of \kspdplus and \svpdplus is significantly lower than the ratio of
\kspdml and \svpdml. While he difference is much smaller for the rest of the road networks, we still observe
that for a very low value of $\theta$, the completeness ratio of the heuristic algorithms that follow
the greedy approach diminishes.

\input{tables/completeness}

%% file: tables/datasets.tex
\begin{table}[t]
  \centering
  \small
  \caption{Road networks tested.}
  \label{tab:datasets}
  \begin{tabular}{ |c|c|c|}
    \hline
    \textbf{Road network} 
    & \textbf{\# of nodes}
    & \textbf{\# of edges} \\
    \hline\hline
    Adlershof		&	349 	 	&	979			\\
    Surat			&	2,508 	 	&	7,398		\\
    Tianjin			&	31,002 		&	86,584		\\ 
    Beijing			&	74,383 	 	&	222,778		\\ 
    \hline 
  \end{tabular}
  
\end{table}

%% file: charts/kdsp-exp-k.tex
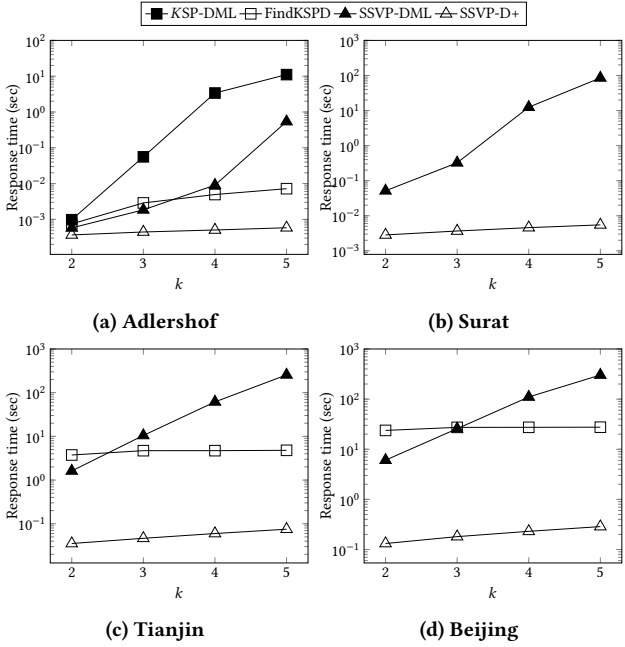
\begin{figure}[t]
	\centering
    \begin{subfigure}[t]{0.23\textwidth}
       	\begin{tikzpicture}[scale=0.5]
       	\tikzstyle{every node}=[font=\LARGE]
           	\begin{axis}[
           		ymode=log,
				xtick={1,2,3,4,5},
       			max space between ticks=30pt,
				try min ticks = 5,
				xlabel={$k$},
    			ylabel={Response time (sec)},
    			ymax=100, ymin=0,
    			y label style={
    				at={(axis description cs:-0.01,0.5)},
    				anchor=north
    			},
    			legend style={
       				at={(1.1,1.18)},
       				anchor=north,
       				legend columns=-1           			
       			},
				mark size=4pt
				]
				\addplot[color=black,mark=square*]   coordinates { 
					(2,0.000981635)		 
					(3,0.0558967) 		 
					(4,3.38668)		 
					(5,11.0779)		 
				};
				\addplot[color=black,mark=square]  coordinates { 
					(2,0.000749625)		 
					(3,0.00288328) 		 
					(4,0.0049503)		 
					(5,0.00720638)		 
				};
				\addplot[color=black,mark=triangle*,mark size=5pt]  coordinates { 
					(2,0.000578637)
					(3,0.00184574)
					(4,0.00905327)
					(5,0.539453)
				};
				\addplot[color=black,mark=triangle,mark size=5pt]  coordinates { 
					(2,0.000368226)
					(3,0.000444847)
					(4,0.000506264)
					(5,0.000583829)	 
				};
				\legend{\kspdml,\kspdplus,\svpdml,\svpdplus}		
			\end{axis}
		\end{tikzpicture}
		\caption{Adlershof}
	\end{subfigure}
    \begin{subfigure}[t]{0.23\textwidth}
		\begin{tikzpicture}[scale=0.5]
        \tikzstyle{every node}=[font=\LARGE]
           	\begin{axis}[
           		ymode=log,
				xtick={1,2,3,4,5},
        		max space between ticks=30pt,
				try min ticks = 5,
				xlabel={$k$},
    			ylabel={Response time (sec)},
    			ymax=1000, ymin=0,
    			y label style={
    				at={(axis description cs:-0.01,0.5)},
    				anchor=north
    			},
				mark size=4pt
				]
				\addplot[color=black,mark=square]  coordinates { 
				};
				\addplot[color=black,mark=triangle*,mark size=5pt]  coordinates { 
					(2,0.051509)
					(3,0.32346)
					(4,12.4591)
					(5,85.1986)	 
				};
				\addplot[color=black,mark=triangle,mark size=5pt]  coordinates { 
					(2,0.00283908)
					(3,0.00367167)
					(4,0.00457822)
					(5,0.00549109)	  
				};
			\end{axis}
		\end{tikzpicture}
		\caption{Surat}
	\end{subfigure}
	\begin{subfigure}[t]{0.23\textwidth}
       	\begin{tikzpicture}[scale=0.5]
       	\tikzstyle{every node}=[font=\LARGE]
        	\begin{axis}[
        		ymode=log,
				xtick={1,2,3,4,5},
       			max space between ticks=30pt,
				try min ticks = 5,
				xlabel={$k$},
   				ylabel={Response time (sec)},
   				ymax=1000, ymin=0,
   				y label style={
   					at={(axis description cs:-0.01,0.5)},
   					anchor=north
   				},
				mark size=4pt
				]
				\addplot[color=black,mark=square]  coordinates { 
					(2,3.74845)
					(3,4.70879)
					(4,4.71122)
					(5,4.79579)	
				};
				\addplot[color=black,mark=triangle*,mark size=5pt]  coordinates { 
					(2,1.61689)
					(3,10.4961)
					(4,61.503)
					(5,255.471)	 
				};
				\addplot[color=black,mark=triangle,mark size=5pt]  coordinates { 
					(2,0.035343)
					(3,0.0464905)
					(4,0.0597991)
					(5,0.0750524)	 
				};
			\end{axis}
		\end{tikzpicture}
		\caption{Tianjin}
	\end{subfigure}
	\begin{subfigure}[t]{0.23\textwidth}
       	\begin{tikzpicture}[scale=0.5]
       	\tikzstyle{every node}=[font=\LARGE]
        	\begin{axis}[
        		ymode=log,
				xtick={1,2,3,4,5},
       			max space between ticks=30pt,
				try min ticks = 5,
				xlabel={$k$},
   				ylabel={Response time (sec)},
   				ymax=1000, ymin=0,
   				y label style={
   					at={(axis description cs:-0.01,0.5)},
   					anchor=north
   				},
				mark size=4pt
				]
				\addplot[color=black,mark=square]  coordinates { 
					(2,23.7525)
					(3,27.3053)
					(4,27.387)
					(5,27.5472)	
				};
				\addplot[color=black,mark=triangle*,mark size=5pt]  coordinates { 
					(2,6.06085)
					(3,25.7291)
					(4,110.3456)
					(5,302.4567)
				};
				\addplot[color=black,mark=triangle,mark size=5pt]  coordinates { 
					(2,0.131794)
					(3,0.18073)
					(4,0.230558)
					(5,0.28661)
				};
			\end{axis}
		\end{tikzpicture}
		\caption{Beijing}
	\end{subfigure}
	\caption{Response time varying requested paths $k$ ($\theta{=}50\%$).}
	\label{fig:kdsp-exp-k}
\end{figure}

%% file: charts/kdsp-exp-theta.tex
\begin{figure}[t]
	\centering
        \begin{subfigure}[t]{0.23\textwidth}
        	\begin{tikzpicture}[scale=0.5]
        	\tikzstyle{every node}=[font=\LARGE]
            	\begin{axis}[
					ymode=log,
					xtick={0.1,0.3,0.5,0.7,0.9},
           			max space between ticks=30pt,
					try min ticks = 5,
					xlabel={$\theta$},
    				ylabel={Response time (sec)},
    				ymax=100, ymin=0,
    				y label style={
    					at={(axis description cs:-0.01,0.5)},
    					anchor=north
    				},
    				legend style={
           				at={(1.1,1.18)},
           				anchor=north,
           				legend columns=-1           			
           			},
					mark size=4pt
				]
				\addplot[color=black,mark=square*]   coordinates { 
					(0.1,10.8737)
					(0.3,1.54688)
					(0.5,0.0558967)
					(0.7,0.00447584)
					(0.9,0.000597405)		 
				};
				\addplot[color=black,mark=square]  coordinates { 
					(0.1,0.0462995)		 
					(0.3,0.0104735) 		 
					(0.5,0.00288328)		 
					(0.7,0.000997884)	
					(0.9,0.00055895)		 
				};
				\addplot[color=black,mark=triangle*,mark size=5pt]  coordinates { 
					(0.1,10.5311)
					(0.3,0.00630321)
					(0.5,0.00184574)
					(0.7,0.000914284)
					(0.9,0.00053954)
				};
				\addplot[color=black,mark=triangle,mark size=5pt]  coordinates { 
					(0.1,0.000729103)
					(0.3,0.00054263)
					(0.5,0.000444847)
					(0.7,0.00038561)
					(0.9,0.000355953)		 
				};
				\legend{\kspdml,\kspdplus,\svpdml,\svpdplus}
				\end{axis}
			\end{tikzpicture}
			\caption{Adlershof}
		\end{subfigure}
        \begin{subfigure}[t]{0.23\textwidth}
			\begin{tikzpicture}[scale=0.5]
        	\tikzstyle{every node}=[font=\LARGE]
            	\begin{axis}[
            		ymode=log,
					xtick={0.1,0.3,0.5,0.7,0.9},
           			max space between ticks=30pt,
					try min ticks = 5,
					xlabel={$\theta$},
    				ylabel={Response time (sec)},
    				ymax=100, ymin=0,
    				y label style={
    					at={(axis description cs:-0.01,0.5)},
    					anchor=north
    				},
					mark size=4pt
				]
				\addplot[color=black,mark=square]  coordinates { 
					(0.1,1.26442)
					(0.3,1.19682)
					(0.5,0.830664)
					(0.7,0.361255)
					(0.9,0.0403501)
				};
				\addplot[color=black,mark=triangle*,mark size=5pt] coordinates { 
					(0.1,6.36858)
					(0.3,1.23355)
					(0.5,0.32346)
					(0.7,0.0691109)
					(0.9,0.00999427)
				};
				\addplot[color=black,mark=triangle,mark size=5pt]  coordinates { 
					(0.1,0.00767064)
					(0.3,0.00550661)
					(0.5,0.00367167)
					(0.7,0.00266654)
					(0.9,0.00220849)		 
				};
				\end{axis}
			\end{tikzpicture}
			\caption{Surat}
		\end{subfigure}
		\begin{subfigure}[t]{0.23\textwidth}
        	\begin{tikzpicture}[scale=0.5]
        	\tikzstyle{every node}=[font=\LARGE]
            	\begin{axis}[
            		ymode=log,
					xtick={0.1,0.3,0.5,0.7,0.9},
           			max space between ticks=30pt,
					try min ticks = 5,
					xlabel={$\theta$},
    				ylabel={Response time (sec)},
    				ymax=1000, ymin=0,
    				y label style={
    					at={(axis description cs:-0.01,0.5)},
    					anchor=north
    				},
    				legend style={
           				at={(1.1,1.25)},
           				anchor=north,
           				legend columns=-1           			
           			},
					mark size=4pt
				]
				\addplot[color=black,mark=square]  coordinates { 
					(0.1,4.95176)
					(0.3,5.08099)
					(0.5,4.70879)
					(0.7,4.00487)
					(0.9,1.0356)
				};
				\addplot[color=black,mark=triangle*,mark size=5pt]  coordinates { 
					(0.1,65.1235)
					(0.3,31.4421)
					(0.5,10.4961)
					(0.7,1.96951)
					(0.9,0.133053)	 
				};
				\addplot[color=black,mark=triangle,mark size=5pt]  coordinates { 
					(0.1,0.140354)
					(0.3,0.0692098)
					(0.5,0.0464905)
					(0.7,0.037639)
					(0.9,0.033773)	 
				};
				\end{axis}
			\end{tikzpicture}
			\caption{Tianjin}
		\end{subfigure}
		\begin{subfigure}[t]{0.23\textwidth}
        	\begin{tikzpicture}[scale=0.5]
        	\tikzstyle{every node}=[font=\LARGE]
            	\begin{axis}[
            		ymode=log,
					xtick={0.1,0.3,0.5,0.7,0.9},
           			max space between ticks=30pt,
					try min ticks = 5,
					xlabel={$\theta$},
    				ylabel={Response time (sec)},
    				ymax=1000, ymin=0,
    				y label style={
    					at={(axis description cs:-0.01,0.5)},
    					anchor=north
    				},
					mark size=4pt
				]
				\addplot[color=black,mark=square]  coordinates { 
					(0.1,26.9248)
					(0.3,26.7375)
					(0.5,27.3053)
					(0.7,24.002)
					(0.9,9.71904)	 
				};
				\addplot[color=black,mark=triangle*,mark size=5pt]  coordinates { 
					(0.1,102.7833)
					(0.3,74.4054)
					(0.5,25.7291)
					(0.7,6.47106)
					(0.9,0.360352) 
				};
				\addplot[color=black,mark=triangle,mark size=5pt]  coordinates { 
					(0.1,0.644113)
					(0.3,0.309239)
					(0.5,0.18073)
					(0.7,0.117807)
					(0.9,0.0967498) 
				};
				\end{axis}
			\end{tikzpicture}
			\caption{Beijing}
		\end{subfigure}
	\caption{Response time varying similarity threshold $\theta$ ($k{=}3$).}
	\label{fig:kdsp-exp-theta}
\end{figure}
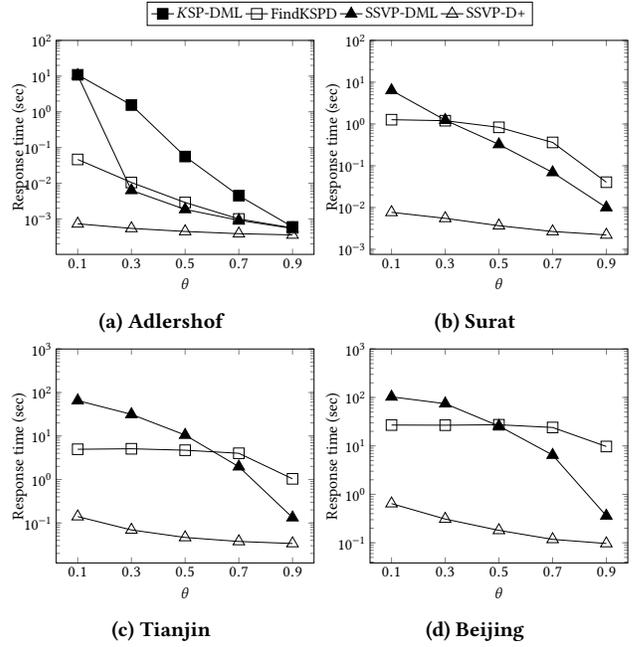

%% file: charts/kdsp-exp-quality.tex
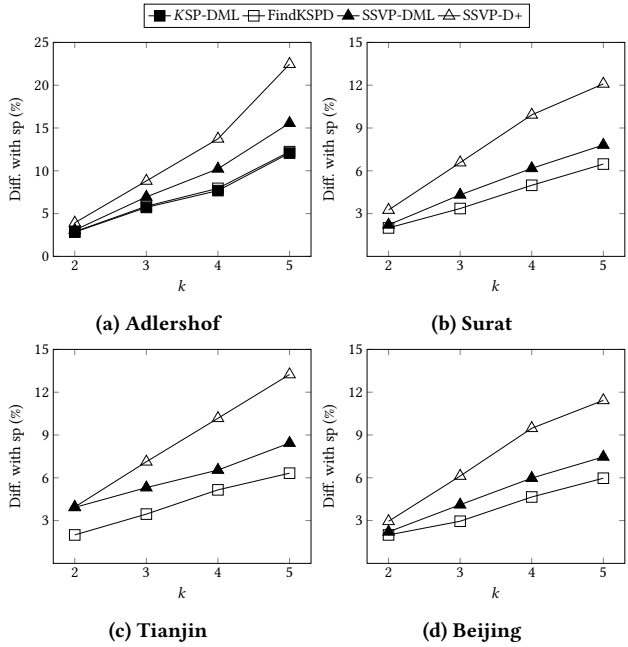
\begin{figure}[t]
  \centering
        \begin{subfigure}[t]{0.23\textwidth}
        	\begin{tikzpicture}[scale=0.5]
        	\tikzstyle{every node}=[font=\LARGE]
            	\begin{axis}[
					xtick={1,2,3,4,5},
           			max space between ticks=30pt,
					try min ticks = 5,
					xlabel={$k$},
    				ylabel={Diff. with sp ($\%$)},
    				ymax=25, ymin=0,
    				y label style={
    					at={(axis description cs:-0.01,0.5)},
    					anchor=north
    				},
    				legend style={
           				at={(1.1,1.18)},
           				anchor=north,
           				legend columns=-1           			
           			},
					mark size=4pt
				]
				\addplot[color=black,mark=square*]   coordinates { 
					(2,2.834)		 
					(3,5.718) 		 
					(4,7.683)		 
					(5,12.064)	 
				};
				\addplot[color=black,mark=square]  coordinates { 
					(2,2.885)		 
					(3,5.854) 		 
					(4,7.947)		 
					(5,12.231)		 
				};
				\addplot[color=black,mark=triangle*,mark size=5pt]  coordinates { 
					(2,3.068)		 
					(3,6.946) 		 
					(4,10.228)		 
					(5,15.564)		 
				};
				\addplot[color=black,mark=triangle,mark size=5pt]  coordinates { 
					(2,3.921)		 
					(3,8.812) 		 
					(4,13.730)		 
					(5,22.459)		 
				};
				\legend{\kspdml,\kspdplus,\svpdml,\svpdplus}	
				\end{axis}
			\end{tikzpicture}
			\caption{Adlershof}
		\end{subfigure}
        \begin{subfigure}[t]{0.23\textwidth}
			\begin{tikzpicture}[scale=0.5]
        	\tikzstyle{every node}=[font=\LARGE]
            	\begin{axis}[
					xtick={1,2,3,4,5},
					ytick={3,6,9,12,15},
           			max space between ticks=30pt,
					try min ticks = 5,
					xlabel={$k$},
    				ylabel={Diff. with sp ($\%$)},
    				ymax=15, ymin=0,
    				y label style={
    					at={(axis description cs:-0.01,0.5)},
    					anchor=north
    				},
    				legend style={
           				at={(1.1,1.25)},
           				anchor=north,
           				legend columns=-1           			
           			},
					mark size=4pt
				]
				\addplot[color=black,mark=square]  coordinates { 
					(2,1.994)		 
					(3,3.354) 		 
					(4,4.983)		 
					(5,6.470)		 
				};
				\addplot[color=black,mark=triangle*,mark size=5pt]  coordinates { 
					(2,2.211)		 
					(3,4.313) 		 
					(4,6.171)		 
					(5,7.808)		 
				};
				\addplot[color=black,mark=triangle,mark size=5pt]  coordinates { 
					(2,3.250)		 
					(3,6.575) 		 
					(4,9.921)		 
					(5,12.092)		 
				};
				\end{axis}
			\end{tikzpicture}
			\caption{Surat}
		\end{subfigure}
        \begin{subfigure}[t]{0.23\textwidth}
			\begin{tikzpicture}[scale=0.5]
        	\tikzstyle{every node}=[font=\LARGE]
            	\begin{axis}[
					xtick={1,2,3,4,5},
					ytick={3,6,9,12,15},
           			max space between ticks=30pt,
					try min ticks = 5,
					xlabel={$k$},
    				ylabel={Diff. with sp ($\%$)},
    				ymax=15, ymin=0,
    				y label style={
    					at={(axis description cs:-0.01,0.5)},
    					anchor=north
    				},
    				legend style={
           				at={(1.1,1.25)},
           				anchor=north,
           				legend columns=-1           			
           			},
					mark size=4pt
				]
				\addplot[color=black,mark=square]  coordinates { 
					(2,1.994)		 
					(3,3.454) 		 
					(4,5.154)		 
					(5,6.323)		 
				};
				\addplot[color=black,mark=triangle*,mark size=5pt]  coordinates { 
					(2,3.923)		 
					(3,5.311) 		 
					(4,6.543)		 
					(5,8.434)		 
				};
				\addplot[color=black,mark=triangle,mark size=5pt]  coordinates { 
					(2,3.950)		 
					(3,7.123) 		 
					(4,10.172)		 
					(5,13.241)		 
				};
				\end{axis}
			\end{tikzpicture}
			\caption{Tianjin}
		\end{subfigure}
        \begin{subfigure}[t]{0.23\textwidth}
			\begin{tikzpicture}[scale=0.5]
        	\tikzstyle{every node}=[font=\LARGE]
            	\begin{axis}[
					xtick={1,2,3,4,5},
					ytick={3,6,9,12,15},
           			max space between ticks=30pt,
					try min ticks = 5,
					xlabel={$k$},
    				ylabel={Diff. with sp ($\%$)},
    				ymax=15, ymin=0,
    				y label style={
    					at={(axis description cs:-0.01,0.5)},
    					anchor=north
    				},
    				legend style={
           				at={(1.1,1.25)},
           				anchor=north,
           				legend columns=-1           			
           			},
					mark size=4pt
				]
				\addplot[color=black,mark=square]  coordinates { 
					(2,1.994)		 
					(3,2.954) 		 
					(4,4.654)		 
					(5,5.970)		 
				};
				\addplot[color=black,mark=triangle*,mark size=5pt]  coordinates { 
					(2,2.211)		 
					(3,4.111) 		 
					(4,5.978)		 
					(5,7.465)		 
				};
				\addplot[color=black,mark=triangle,mark size=5pt]  coordinates { 
					(2,2.950)		 
					(3,6.123) 		 
					(4,9.472)		 
					(5,11.428)		 
				};
				\end{axis}
			\end{tikzpicture}
			\caption{Beijing}
		\end{subfigure}
	\caption{Result quality varying requested paths $k$ ($\theta{=}50\%$).}
  \label{fig:kdsp-exp-quality}
\end{figure}

%% file: tables/completeness.tex
\begin{table}[htb]
  	\centering
  	\small
  	\caption{Percentage of queries for which a complete result set is found varying requested paths $k$ ($\theta{=}50\%$). }
  	\label{tab:completeness}
	\def\arraystretch{0.9}
\scalebox{0.9}{%
	\begin{tabular}{|c|c|c||c|c|c|c|}
          \hline
          \textbf{Road net.}
          & \textbf{$k$}
          & \textbf{$\theta$}
          & \textbf{\kspdml}
          & \textbf{\kspdplus}
          & \textbf{\svpdml}
          & \textbf{\svpdplus}
          \\\hline\hline
          \multirow{8}{*}{Adlershof}
          & 2	&	0.5	&	100		&	100		&	100		&	99.9	\\
          & 3	&	0.5	&	100		&	100	 	&	100		&	99.7	\\
          & 4	&	0.5	&	100		&	100		&	99.9	&	99.2	\\
          & 5	&	0.5	&	100		&	100		&	99.9	&	98.5	\\
          & 3	&	0.9	&	100		&	100		&	100		&	99.9	\\
          & 3	&	0.7	&	100		&	100		&	100		&	99.9	\\	
          & 3	&	0.3	&	100		&	99.1	&	99.8	&	94.7	\\
          & 3	&	0.1	&	91.8	&	79.7	&	91.7	&	56.1	\\
          \hline
          \multirow{8}{*}{Surat}
          & 2	&	0.5	&	-		&	100		&	100		&	100		\\
          & 3	&	0.5	&	-		&	100		&	100		&	99.9	\\
          & 4	&	0.5	&	-		&	100		&	99.9	&	99.9	\\
          & 5	&	0.5	&	-		&	99.9	&	99.9	&	99.5	\\
          & 3	&	0.9	&	-		&	100		&	100		&	99.9	\\
          & 3	&	0.7	&	-		&	100		&	100		&	99.9	\\	
          & 3	&	0.3	&	-		&	99.9	&	99.9	&	99.4	\\
          & 3	&	0.1	&	-		&	96.4	&	98.1	&	83.0	\\
          \hline
          \multirow{8}{*}{Tianjin}
          & 2	&	0.5	&	-		&	100		&	100		&	100		\\
          & 3	&	0.5	&	-		&	100		&	100		&	100		\\
          & 4	&	0.5	&	-		&	100		&	100		&	99.9		\\
          & 5	&	0.5	&	-		&	100		&	100		&	99.6		\\
          & 3	&	0.9	&	-		&	100		&	100		&	100		\\
          & 3	&	0.7	&	-		&	100		&	100		&	100		\\	
          & 3	&	0.3	&	-		&	99.9	&	99.8	&	99.0		\\
          & 3	&	0.1	&	-		&	99.8	&	96.3	&	88.6		\\
          \hline
          \multirow{8}{*}{Beijing}
          & 2	&	0.5	&	-		&	100		&	100		&	100		\\
          & 3	&	0.5	&	-		&	100		&	100		&	100		\\
          & 4	&	0.5	&	-		&	100		&	100		&	99.9		\\
          & 5	&	0.5	&	-		&	100		&	100		&	99.8		\\
          & 3	&	0.9	&	-		&	100		&	100		&	100		\\
          & 3	&	0.7	&	-		&	100		&	100		&	100		\\	
          & 3	&	0.3	&	-		&	99.9	&	99.8	&	99.0		\\
          & 3	&	0.1	&	-		&	99.8	&	96.1	&	90.8	\\
          \hline
	\end{tabular}	
}
\end{table}

%% file: chapters/conclusion.tex
\section{Conclusions}\label{sec:conclusion}

In this paper, we studied the \kdsp problem, which aims at computing $k$ dissimilar
paths while minimizing their collective length. We showed that the problem is
strongly $NP$-hard.  We also presented an exact algorithm, that iterates over all
paths from $s$ to $t$ in length order, along with two pruning criteria to reduce
the number of examined paths. As iterating over all paths from $s$ to $t$ is impractical,
we introduced the much smaller set of simple single-via paths, and we presented two
heuristic algorithms that iterate over this much smaller set to process \kdsp queries.
Our experiments showed that iterating over the set of simple single-via paths in length
order instead of all paths from $s$ to $t$ can lead to scalable solutions with a small
trade-off in the quality of the results.

%% file: tchond-corr-2018.bbl
\begin{thebibliography}{10}

\bibitem{cambridge2005}
{Choice Routing}.
\newblock Cambridge Vehicle Information Technology Ltd., 2005.

\bibitem{abraham2013}
Ittai Abraham, Daniel Delling, Andrew~V Goldberg, and Renato~F Werneck.
\newblock {Alternative Routes in Road Networks}.
\newblock {\em Journal of Experimental Algorithmics}, 18:1--17, 2013.

\bibitem{akgun2000}
Vedat Akgun, Erhan Erkut, and Rajan Batta.
\newblock {On Finding Dissimilar Paths}.
\newblock {\em European Journal of Operational Research}, 121(2):232--246,
  2000.

\bibitem{arora2009}
Sanjeev Arora and Boaz Barak.
\newblock {\em Computational Complexity: A Modern Approach}, chapter 2.4.
\newblock Cambridge University Press, 2009.

\bibitem{bader2011}
Roland Bader, Jonathan Dees, Robert Geisberger, and Peter Sanders.
\newblock {Alternative Route Graphs in Road Networks}.
\newblock In {\em Proc. of the 1st Int. ICST Conf. on Theory and Practice of
  Algorithms in (Computer) Systems}, pages 21--32, 2011.

\bibitem{chondrogiannis2015}
Theodoros Chondrogiannis, Panagiotis Bouros, Johann Gamper, and Ulf Leser.
\newblock {Alternative Routing: K-shortest Paths with Limited Overlap}.
\newblock In {\em Proc. of the 23rd ACM SIGSPATIAL Conf.}, pages 68:1--68:4,
  2015.

\bibitem{chondrogiannis2017}
Theodoros Chondrogiannis, Panagiotis Bouros, Johann Gamper, and Ulf Leser.
\newblock {Exact and Approximate Algorithms for Finding $k$-Shortest Paths with
  Limited Overlap}.
\newblock In {\em Proc. of the 20th EDBT Conf.}, pages 414--425, 2017.

\bibitem{delling2009}
Daniel Delling and Wagner Dorothea.
\newblock {Pareto Paths with SHARC}.
\newblock In {\em Proc. of the 8th Int. Symposium on Experimental Algorithms},
  pages 125--136, 2009.

\bibitem{eilam-tzoreff1998}
Tali Eilam-Tzoreff.
\newblock The disjoint shortest paths problem.
\newblock {\em Discrete applied mathematics}, 85(2):113--138, 1998.

\bibitem{jeong2009}
Yeon-Jeong Jeong, Tschangho~John Kim, Chang-Ho Park, and Dong-Kyu Kim.
\newblock {A Dissimilar Alternative Paths-search Algorithm for Navigation
  Services: A Heuristic Approach}.
\newblock {\em KSCE Journal of Civil Engineering}, 14(1):41--49, 2009.

\bibitem{karduni2016}
Alireza Karduni, Amirhassan Kermanshah, and Sybil Derrible.
\newblock {A Protocol to Convert Spatial Polyline Data to Network Formats and
  Applications to World Urban Road Networks}.
\newblock {\em Scientific Data}, 3(160046), 2016.

\bibitem{knuthTAOCP}
Donald~E. Knuth.
\newblock {\em The Art of Computer Programming, Vol. 4, Fas. 3: Generating All
  Combinations and Partitions}.
\newblock Addison-Wesley Professional, Boston, MA, USA, 2005.

\bibitem{kriegel2010}
Hans-Peter Kriegel, Matthias Renz, and Matthias Schubert.
\newblock {Route skyline queries: A multi-preference path planning approach}.
\newblock In {\em Proc. of the 26th IEEE ICDE}, pages 261--272, 2010.

\bibitem{lim2005}
Yongtaek Lim and Hyunmyung Kim.
\newblock {A Shortest Path Algorith for Real Road Network based on Path
  Overlap}.
\newblock {\em Journal of the Eastern Asia Society for Transportation Studies},
  6:1426--1438, 2005.

\bibitem{liu2017}
H~Liu, C.~Jin, B~Yang, and A.~Zhou.
\newblock Finding top-k shortest paths with diversity.
\newblock {\em IEEE TKDE}, 30(3):488--502, 2017.

\bibitem{luxen2015}
Dennis Luxen and Dennis Schieferdecker.
\newblock {Candidate Sets for Alternative Routes in Road Networks}.
\newblock {\em Journal of Experimental Algorithmics}, 19:1--28, 2015.

\bibitem{mouratidis2010}
Kyriakos Mouratidis, Yimin Lin, and Man Lu~Yiu.
\newblock {Preference Queries in Large Multi-cost Transportation Networks}.
\newblock In {\em Proc. of the 26th IEEE ICDE}, pages 533--544, 2010.

\bibitem{roberts2007}
Ben Roberts and Dirk~P Kroese.
\newblock Estimating the number of st paths in a graph.
\newblock {\em J. Graph Algorithms Appl.}, 11(1):195--214, 2007.

\bibitem{samet2008}
Hanan Samet, Jagan Sankaranarayanan, and Houman Alborzi.
\newblock {Scalable Network Distance Browsing in Spatial Databases}.
\newblock In {\em Proc. of the 2008 ACM SIGMOD Conf.}, pages 43--54, 2008.

\bibitem{shekelyan2015}
Michael Shekelyan, Gregor Joss{\'{e}}, and Matthias Schubert.
\newblock {Linear path skylines in multicriteria networks}.
\newblock In {\em Proceedings of the 31st IEEE International Conference on Data
  Engineering}, ICDE'15, pages 459--470, Washington, DC, USA, 2015. IEEE.

\bibitem{xie2012}
Kexin Xie, Ke~Deng, Shuo Shang, Xiaofang Zhou, and Kai Zheng.
\newblock {Finding Alternative Shortest Paths in Spatial Networks}.
\newblock {\em ACM Transactions on Database Systems}, 37(4):29:1--29:31, 2012.

\bibitem{yen1971}
Jin~Y Yen.
\newblock {Finding the K Shortest Loopless Paths in a Network}.
\newblock {\em Management Science}, 17(11):712--716, 1971.

\end{thebibliography}
